%% file: o-main.tex
\documentclass[a4paper,english, cleveref, autoref, thm-restate, nolineno]{socg-lipics-v2021}
\usepackage[ruled,linesnumbered,vlined]{algorithm2e}
\usepackage{wrapfig}

\graphicspath{{./figures/}}

\newcommand{\poly}{{\rm poly}}

\def\A{\mathcal{A}}

\def\D{\mathcal{D}}
\def\F{\mathcal{F}}

\def\K{\mathcal{K}}

\def\reals{{\mathbb R}}

\def\eps{{\varepsilon}}

\newcommand{\argmin}{\rm argmin}

\newcommand{\oldparagraph}[1]{\vspace{5pt}\noindent\textbf{#1}}

\def\VP{\text{Vis}}        
\def\OPT{\text{OPT}}

\definecolor{forestgreen}{rgb}{0.13, 0.55, 0.13}
\def\dfn#1{\emph{\textcolor{forestgreen}{\textbf{#1}}}}

\newcommand{\old}[1]{{{}}}
\newcommand{\temp}[1]{{{}}}
\newcommand{\fullversion}[1]{{{}}}

\usepackage[colorinlistoftodos,prependcaption,textsize=small]{todonotes}

\newcommand{\omrit}[1]{}
\newcommand{\omritin}[1]{}
\newcommand{\matya}[1]{}
\newcommand{\matyain}[1]{}
\newcommand{\joe}[1]{}
\newcommand{\joein}[1]{}
\newcommand{\rathish}[1]{}
\newcommand{\rathishin}[1]{}
\newcommand{\todoin}[1]{}
\newcommand{\rev}[1]{}

\title{Robustly Guarding Polygons}

\author{Rathish Das}{University of Houston}{rathish@central.uh.edu}{https://orcid.org/0000-0002-2416-6422}{}
\author{Omrit Filtser}{The Open University of Israel}{omrit.filtser@gmail.com}{https://orcid.org/0000-0002-3978-1428}{}
\author{Matthew J. Katz}{Ben-Gurion University of the Negev}{matya@cs.bgu.ac.il}{https://orcid.org/0000-0002-0672-729X}{Partially supported by Grant 495/23 from the Israel Science Foundation and Grant 2019715/CCF-20-08551 from the US-Israel Binational Science Foundation/US National Science Foundation.}
\author{Joseph S.B. Mitchell}{Stony Brook University}{joseph.mitchell@stonybrook.edu}{https://orcid.org/0000-0002-0152-2279}{Partially supported by the National Science Foundation (CCF-2007275).}
\authorrunning{R. Das, O. Filtser, M.\,J. Katz, and J.\,S.B. Mitchell} %TODO mandatory. First: Use abbreviated first/middle names. Second (only in severe cases): Use first author plus 'et al.'

\Copyright{Rathish Das, Omrit Filtser, Matthew J. Katz, and Joseph S.B. Mitchell} %TODO mandatory, please use full first names. LIPIcs license is "CC-BY";  http://creativecommons.org/licenses/by/3.0/

 %TODO mandatory; please add comma-separated list of keywords

\relatedversion{} %optional, e.g. full version hosted on arXiv, HAL, or other respository/website
\EventEditors{Wolfgang Mulzer and Jeff M. Phillips}
\EventNoEds{2}
\EventLongTitle{40th International Symposium on Computational Geometry (SoCG 2024)}
\EventShortTitle{SoCG 2024}
\EventAcronym{SoCG}
\EventYear{2024}
\EventDate{June 11-14, 2024}
\EventLocation{Athens, Greece}
\EventLogo{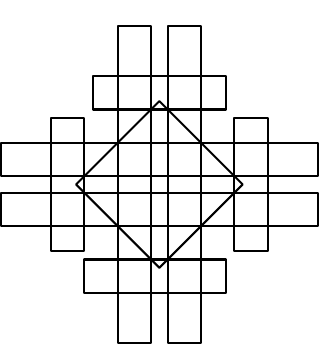}
\SeriesVolume{293}
\ArticleNo{39}
\begin{document}

\maketitle

\begin{abstract}
       We propose precise notions of what it means to guard a domain ``robustly'', under a variety of models. While approximation algorithms for minimizing the number of (precise) point guards in a polygon is a notoriously challenging area of investigation, we show that imposing various degrees of robustness on the notion of visibility coverage leads to a more tractable (and realistic) problem for which we can provide approximation algorithms with constant factor guarantees. 
\end{abstract}

\input{o-sec_intro}

\input{o-sec_robust_guarding_definitions}

\input{o-sec_const_approx}

\input{sec_improved_approx}

\bibliography{refs}

\appendix

\input{o-sec_missing_proofs}

\end{document}

%% file: o-sec_intro.tex
\section{Introduction}
A fundamental set cover problem that arises in geometric domains is the classic ``art gallery'' or ``guarding'' problem: Given a geometric domain (e.g., polygon $P$), place a set of points (``guards'') within $P$, such that every point of $P$ is seen by at least one of the guards. This problem has many variants and has been studied extensively from many perspectives, including combinatorics, complexity, approximation algorithms, and algorithm engineering for solving real instances to provable optimality or near-optimality.

Approximation algorithms for guarding have been extensively pursued for decades (see related work), where the various variants differ from one another in (i) the underlying domain, e.g., simple polygon vs. polygon with holes, (ii) the portion of the domain that must be guarded, e.g., only its boundary, the entire domain, or a discrete set of points in it (iii) the type of guards, e.g., static, mobile, or with various restrictions on their coverage area, (iv) the restrictions on the location of the guards, e.g., only at vertices (vertex guards) or anywhere (point-guards), and (v) the underlying notion of vision, e.g., line of sight or rectangle visibility. 
Despite all this work, even an $O(\log \OPT)$-approximation algorithm for point-guarding a simple $n$-gon, without (weak) additional assumptions, is unknown (see below). 

In this paper, we present and discuss a new and natural notion of vision called \emph{robust vision}. Under the standard notion of vision, two points in a polygon $P$ see each other if and only if the line segment between them is contained in $P$. However, in the context of guarding, where e.g. the guards are not necessarily stationary entities or perhaps their location is imprecise, it makes sense to require that a guard $g$ that is responsible for guarding a point $p$ can see $p$ from any point in some vicinity of its specified location.  In this case we say that $g$ robustly guards $p$. It is also possible that the location of an entity to be guarded is imprecise or alternatively the entity may move in the vicinity of its specified location, and we would like to ensure that the guard in charge of this entity does not lose sight of it, i.e., we would like it to guard the entity robustly. (Here, we mostly focus on the former case of robust vision.) 
Robust guarding is a generalization of standard guarding in that when the requirement of robustness tends to zero, robust guarding reduces to standard guarding.

Robust guarding is especially important in light of recent results showing that there are polygons that can be guarded with 2 guards, but only if both guards are very precisely placed at points with irrational coordinates~\cite{meijer2022sometimes} (see also \cite{abrahamsen2017irrational}). In our formulation of robust guarding, we explicitly model the fact that a guard may be imprecisely placed at locations within a polygon, and may in fact move around within some neighborhood; we insist, then, for a point $p$ to be ``seen'' that it must be seen no matter where the guard may be within a disk that subtends at least some minimum angle when viewed from $p$. The breakthrough result~\cite{abrahamsen2021art} showing that the guarding problem is complete within the existential theory of the reals is strong evidence of the algebraic difficulty of computing exact optimal sets of guards, even in polygons in the plane.

\oldparagraph{Our techniques and results.} We summarize our contributions and methods:

(1) We introduce notions of ``robust vision'' within a polygonal domain $P$ and analyze the optimal guarding problem from this new perspective. In particular, for an appropriately small parameter $\alpha>0$, we say that a guard at point $g$ \emph{$\alpha$-robustly guards} a point $p\in P$ if $p$ sees (under ordinary visibility) all points within a Euclidean disk of radius $\alpha\|g-p\|$ centered at $g$. In the figure below, $g$ $\alpha$-robustly guards $p$, but not $p'$ or $p''$. 

\begin{figure}[h]
	\centering
	\includegraphics{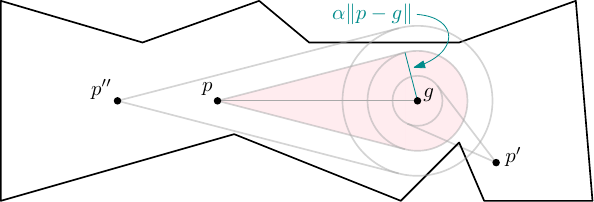}
\end{figure}

Note that as $\alpha$ approaches $0$, the degree of robustness decreases, and at the limit we get standard guarding, where $g$ guards $p$ if and only if $\overline{gp}\subset P$.
\matya{slightly elaborated}
We characterize the $\alpha$-robust visibility region, $\VP_\alpha(g)$, of all points $\alpha$-robustly visible from $g$ (\Cref{sec:robust-visibility-region}), as well as the region $\VP^{\text{inv}}_\alpha(p)$ of all points $g$ from which $p$ is $\alpha$-robustly visible (\Cref{sec:inv_vis}). In particular, we prove that $\VP_\alpha(g)$, which in general is not a polygon, is star-shaped and $O(\alpha)$-fat. Moreover, we show that both regions can be computed efficiently.

(2) We show that, as with ordinary guarding, the problem of computing a minimum cardinality set of guards in $P$ that $\alpha$-robustly see all of $P$ is APX-hard, making it unlikely that there exists a PTAS or an exact polynomial-time algorithm for the problem.

(3) We present an $O(1)$-approximation algorithm for robustly guarding a general polygonal domain $P$. (The approximation factor depends on the robustness parameter $\alpha$, and the algorithm is a bicriteria, allowing a slight relaxation of $\alpha$.) This is to be contrasted with the situation for ordinary guarding, for which even finding a logarithmic-factor approximation algorithm for placing guards at points within a simple polygon requires some additional (weak) assumptions.

Specifically, \Cref{thm:robust-greedy-approx} states that, given a polygon $P$ with $n$ vertices, one can compute in $\poly(n)$ time the cardinality of, and an implicit representation of, a set of $O(\alpha^{-6})|\OPT_\alpha|$ points that $\alpha/8$-robustly guard $P$, where $\OPT_\alpha$ is a minimum-cardinality set of guards that $\alpha$-robustly guard $P$. In additional time $O(|\OPT_\alpha|)$ one can output an explicit set of such points.
We present this result by first presenting a result of a similar flavor for robustly guarding a discrete set $S$ of points within $P$ (\Cref{thm:robust-greedy-approx-discrete}).

Critical to our main result (\Cref{thm:robust-greedy-approx}) is \Cref{thm:robust-vision-candidates}, which shows that one can compute a discrete set $Q$ of points (candidate guards) that is guaranteed to contain a subset of $O(\alpha^{-4})|OPT_\alpha|$ points that $\alpha/4$-robustly guard $P$.
This result is quite delicate and requires some technical geometric analysis, utilizing a medial axis decomposition and carefully placed grid points in portions of $P$.  This is in contrast with the classic guarding problem in which the existence of a polynomial-size discrete candidate set that suffices for good approximation has been elusive.

(4) In \Cref{sec:smaller-solution} we extend/generalize our definition of robust vision to include the possibility that $p$ need not see \emph{all} of the neighborhood of $g$, but only a \emph{fraction} of that neighborhood, and we require that $g$ sees a fraction of a neighborhood of $p$ as well. Within this model we are able to obtain improved factors, with some tradeoffs; see \Cref{thm:refined-robust}.

We defer many proofs and technical arguments to the appendix, while attempting to convey intuition, and detailed figures, in the body of the paper.

\oldparagraph{Related Work.}
Eidenbenz, Stamm, and Widmayer~\cite{eidenbenz2001inapproximability} have shown that optimally guarding a simple polygon $P$ is not only NP-hard (a classical result~\cite{lee1986computational,o2017visibility,shermer1992recent}) but is APX-hard (there is no PTAS unless P=NP); if $P$ has holes, they show that there is no $o(\log n)$ approximation algorithm unless $P=NP$. 
A recent breakthrough of Abrahamsen, Adamaszek and Miltzow~\cite{abrahamsen2021art} has shown that point guarding in general polygons is $\exists \mathbb{R}$-complete, making it unlikely the problem is in NP. Further, \cite{abrahamsen2017irrational,meijer2022sometimes} have shown that optimal solutions to even very small problems requiring 2-3 guards in simple polygons may require precise placement of guards at irrational points.\footnote{It is an interesting open problem to determine whether or not an optimal set of robust guards might require irrational coordinates for guards, for input polygons $P$ with integer coordinates.}
These results imply the necessity of algebraic methods to compute exact solutions.
Efrat and Har-Peled~\cite{EH06} present a randomized $O(\log \OPT_{\text{grid}})$-approximation algorithm where the placement of guards is restricted to a fine grid. However, they do not prove that their approximation ratio holds when compared to general point guard placement of optimal solution.
Building on \cite{EH06} and on Deshpande et al.~\cite{DKDS07}, Bonnet and Miltzow~\cite{bonnet2017approximation} gave a randomized $O(\log \OPT)$-approximation algorithm for point guards within a simple polygon $P$ under mild assumptions: vertices have integer coordinates, no three vertices are collinear, and no three extensions meet in a point within $P$
that is not a vertex, where an extension is a line passing through two vertices of $P$. 
The problem has also been examined from the perspective of smoothed analysis~\cite{dobbins2018smoothed,erickson2022smoothing} and parameterized complexity~\cite{ashok2019efficient,agrawal2020parameterized,bonnet2020parameterized,agrawal2020parameterized-socg}.
For guards that must be placed at discrete locations on the boundary of a simple polygon $P$, King and Kirkpatrick~\cite{KK11,kirkpatrick2015lg}
obtained an $O(\log\log OPT)$-approximation, by building $\eps$-nets
of size $O((1/\eps)\log\log (1/\eps))$ for the associated hitting set
instances, and applying~\cite{BronnimannG95}. If the disks bounding the visibility polygons at these discrete locations are shallow (i.e., every point in the domain is covered by $O(1)$ disks), then a local search based PTAS exists~\cite{AschnerKMY13}.

For simple polygons with special structures, such as monotone polygons, terrains, and weakly-visible polygons~\cite{krohn2013approximate,Ben-MosheKM07,krohn2014guarding,FHKS16,AshurFKS22,BGR17,AshurFK21}, there are improved approximation ratios (constant, or even PTAS), but these algorithms utilize the very special structures of these classes of polygons.
Aloupis et al~\cite{AloupisBDGLS14} considered guarding in ``$(\alpha,\beta)$-covered'' 
polygons (which are intuitively ``fat'' polygons, see~\cite{Efrat05}), and showed that the boundary of such polygons can be guarded by $O(1)$ guards.
Polygons in which no point sees a particularly small area also have special analysis and algorithms~\cite{kalai1997guarding,kirkpatrick2000guarding,valtr1998guarding,valtr1999galleries}.

Some prior work has addressed guarding from a robustness perspective, though the perspectives are significantly different from ours. Efrat, Har-Peled, and Mitchell~\cite{EfratHM05} consider a definition of robust guarding in which a point is robustly guarded if it is seen by at least 2 guards from significantly different angles.
Hengeveld and Miltzow~\cite{hengeveld2021practical} consider a notion of ``robust'' vision by examining the impact that certain changes to ``visibility'' has on the optimal number of guards needed to cover a domain $P$. They introduce the notion of \emph{vision-stability}: A polygon $P$ has vision-stability $\delta$ if the optimal number of ``enhanced guards'' (who can see an additional angle $\delta$ around corners) is equal to the optimal number of ``diminished guards'' (whose visibility region is decreased by an angle $\delta$ at each shadow-casting corner).  

Practical methods employing heuristics, algorithm engineering and combinatorial optimization have been successful for computing exact or approximately optimal solutions in many instances~\cite{amit2010locating, borrmann2013point,hengeveld2021practical,tozoni2016algorithm}, though some of the exact methods can potentially fail or run forever on contrived instances, e.g., those requiring irrational guards. (\cite{hengeveld2021practical} is able to detect such contrived instances and to solve exactly instances that are vision-stable.)

%% file: o-sec_robust_guarding_definitions.tex
\section{Robust guarding}\label{sec:robust-guarding}
In this section we introduce and discuss a new and natural notion of vision called \emph{robust vision}. Let $P$ be a polygonal domain (a polygon, possibly with holes) in the plane having in total $n$ vertices. Under the standard notion of vision, a point $p\in P$ sees another point $q\in P$ if the segment $\overline{pq}$ is contained in $P$. In the context of guarding, where the guards are not necessarily stationary or their locations are imprecise,
one would like to ensure that a guard does not lose eye contact with any of the points it is responsible to see, when it moves in some vicinity of its specified location. 
The following definition attempts to capture these common situations. Here, we denote by $D(p,r)$ the disk of radius $r$ centered at a point $p$.

\begin{definition}[Robust Guarding]\label{def:robust_visibility}
Given a polygon $P$ and parameter $0<\alpha\le 1$, we say that a point $g\in P$ \dfn{$\alpha$-robustly guards} a point $p\in P$ if $p$ sees $D(g,\alpha\cdot\|p-g\|)$. 
\end{definition}

\begin{figure}[h]
	\centering
	\includegraphics{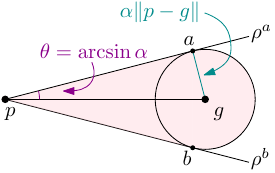}
	\caption{\label{fig:visibility_cone} A point $g$ that $\alpha$-robustly guards another point $p$. The pink ``ice cream cone'' is contained in the polygon $P$.}
\end{figure}

Let $g,p$ be two points in a polygon $P$, such that $g$ $\alpha$-robustly guards $p$. By definition, $p$ sees the disk $D(g,\alpha\|p-g\|)$. Consider the two rays $\rho^a,\rho^b$ from $p$ tangent to $D(g,\alpha\|p-g\|)$. Let $a$ (resp. $b$) be the tangency point of $\rho^a$ (resp. $\rho^b$) and $D(g,\alpha\|p-g\|)$. Notice that the segments $\overline{ap},\overline{bp}$ with the disk $D(g,\alpha\|p-g\|)$ create an ``ice cream cone''. 
Formally, we define the \dfn{ice cream cone} from $p$ to $g$ as the union of the triangle $\triangle apb$ and the disk $D(g,\alpha\|p-g\|)$. Since $p$ must see the entire ice cream cone, we have the following observation.

\begin{observation}\label{obs:ice-cream-cone-in-P}
     A point $g\in P$ $\alpha$-robustly guards a point $p\in P$ if and only if the ice cream cone from $p$ to $g$ is contained in $P$.
\end{observation}

Consider the triangle $\triangle agp$. We have $\|a-g\|=\alpha\|p-g\|$ and thus $\angle apg=\arcsin\alpha=\theta$. Symmetrically, $\angle bpg=\theta$, and we get $\angle apb=2\theta$. We obtain the following observation.
\begin{observation}\label{obs:ice-cream-cone-angle}
    The angle $\angle apb$ of the ice cream cone from $p$ to $g$ is $2\theta=2\arcsin\alpha$.
\end{observation}

Note that as $\alpha$ approaches $0$, the degree of robustness decreases, and at the limit we get standard guarding, where $g$ guards $p$ if and only if $\overline{gp}\subset P$.
On the other hand, if $P$ has a vertex $v$ with internal angle smaller than $2\theta$, then $v$ cannot be $\alpha$-robustly guarded by any other point $g\ne v$ in $P$ ($v$ can guard itself). Nevertheless, if all internal angles in $P$ are at least $2\theta$, then there is always a finite set of guards $G$ that robustly guard $P$ (this will become clear in \Cref{sec:robust-guarding-algorithm}). We therefore assume that $\alpha\le \sin \frac{\phi}{2}$, where $\phi$ is the smallest internal angle of $P$, and thus $2\theta=2\arcsin\alpha\le\phi$, so $P$ can be guarded by a finite number of $\alpha$-robust guards. Of course, if $\phi$ is very `small', then, as mentioned above, the degree of robustness decreases, in the sense that the vicinity in which a guard may move while guarding a point $p$ becomes more limited.  

Also note that $P$ does not have to be fat in order to be guarded robustly; however, the number of robust guards required may depend on geometric features of $P$ (see \Cref{fig:thin_rectangle}).

\begin{figure}[h]
	\centering
	\includegraphics[scale=1.2]{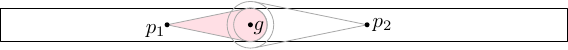}
	\caption{\label{fig:thin_rectangle} In a thin rectangle $P$, the point $g$ $\alpha$-robustly guards $p_1$, but not $p_2$. Here, the number of robust guards required depends on the aspect ratio of the rectangle $P$.}
\end{figure}

The geometric observation below will be very useful in the following sections. A proof is given in \Cref{sec:proofs-robust-guarding}, see \Cref{fig:disk_in_cone} for an illustration.
\begin{restatable}{observation}{diskInCone}\label{obs:disk_in_cone}
    Consider a (convex) cone $K$ defined by two rays $\rho_0,\rho_1$ from a point $p$, such that the small angle between them is $\theta=\arcsin\alpha$. Let $q$ be a point in $K$, such that the smaller angle between $\overline{pq}$ and one of the rays $\rho_0,\rho_1$ is $c\cdot\theta$ (for some $c<1$). Then the disk $D(q,c\alpha\|p-q\|)$ is contained in $K$.
\end{restatable}
\begin{figure}[h]
	\centering
	\includegraphics[scale=1.3]{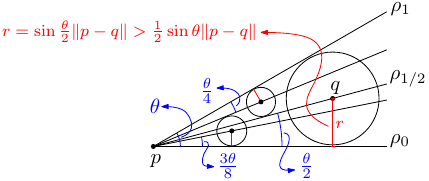}
	\caption{\label{fig:disk_in_cone} The cone $K$ with angle $\theta$. For $q$ on $\rho_{1/2}$, the disk $D(q,\frac12\alpha\|p-q\|)$ is contained in $K$.}
\end{figure}
\subsection{The robust visibility region}\label{sec:robust-visibility-region}
Let $P$ be a polygon in the plane, possibly with holes, having in total $n$ vertices. Denote by $\VP(p)$ the standard visibility polygon of $p$, i.e., $\VP(p)$ is the set of points $q \in P$ such that $q$ is visible from $p$. Let $\VP_\alpha(p)$ denote the \dfn{$\alpha$-robust visibility region} of $p$, i.e. the set of points in $P$ that are $\alpha$-robustly guarded by $p$. Interestingly, unlike $\VP(p)$, the robust visibility region is rarely a polygon. Nonetheless, in this subsection we present an efficient algorithm to compute it. First, we characterize $\VP_\alpha(g)$ and reveal some of its interesting geometric properties. We begin by showing that $\VP_\alpha(g)$ is fat, under the following standard definition of fatness.

\oldparagraph{Fatness.}
For a disk $D(p,r)$ of radius $r$ centered at point $p\in P$, let $C(p,r)$ denote the connected component of $D(p,r)\cap P$ in which $p$ lies. We say that a polygon $P$ is \dfn{$\gamma$-fat} if for any point $p\in P$ and radius $r$ such that $D(p,r)$ does not fully contain $P$, the area of $C(p,r)$ is at least $\gamma\cdot \pi r^2$, i.e., at least $\gamma$ times the area of $D(p,r)$.

\oldparagraph{}
We use \Cref{clm:union_kites_fat} below to show that $\VP_\alpha(g)$ is star-shaped and $O(\alpha)$-fat. In order not to interrupt the flow of reading with rather technical details, we defer its proof to \Cref{sec:proofs_robust_vis}.

\begin{restatable}{claim}{unionKitesFat}\label{clm:union_kites_fat}
    Let $\K$ be a set of (convex) $\gamma$-fat kites (quadrilaterals with reflection symmetry across a diagonal), all having a common point $p$. Then the union $U=\bigcup_{K\in\K}K$ is $\gamma/4$-fat.
\end{restatable}

\begin{figure}[h]
	\centering
	\includegraphics[scale=1]{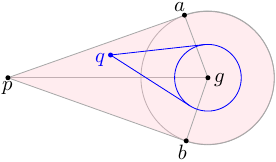}
	\caption{\label{fig:VP_1_0_fat} The ice cream cone $C$ from $p$ to $g$ is shaded in pink. For any point $q \in C$, the ice cream cone from $q$ to $g$ is contained in $C$. Therefore, $\VP_\alpha(g)$ contains $C$.}
\end{figure}
\begin{lemma}\label{clm:VP_1_0_fat}
    For any $0<\alpha < 1$ and $g\in P$, $\VP_\alpha(g)$ is star-shaped and $O(\alpha)$-fat.
\end{lemma}
\begin{proof}
Let $p$ be a point in $\VP_\alpha(g)$, and denote by $C$ the ice cream cone from $p$ to $g$. Notice that for any point $q\in C$ we have $\|q-g\|\le \|p-g\|$, and thus the ice cream cone from $q$ to $g$ is contained in $C$ (see \Cref{fig:VP_1_0_fat}). 
 Since $C$ is contained in $P$, $q$ sees a disk of radius $\alpha\|q-g\|$ centered at $g$, and thus $q\in \VP_\alpha(g)$. We conclude that $C\subseteq \VP_\alpha(g)$, and that $\VP_\alpha(g)$ is star-shaped.

Let $a,b$ be the tangency points defining the cone $C$. We have $\angle pag=\angle pbg=\frac{\pi}{2}$, $\angle apb=2\theta$, and $\angle agb=\pi-2\theta$, and thus the kite $pagb$ is $O(\alpha)$-fat. Since $C\subseteq \VP_\alpha(g)$, the union of kites over all $p\in \VP_\alpha(g)$ is exactly $\VP_\alpha(g)$. By applying \Cref{clm:union_kites_fat} on the set of kites (all have a common vertex $g$), we conclude that $\VP_\alpha(g)$ is $O(\alpha)$-fat.  
\end{proof}

Given a point $g\in P$, denote by $D_g$ the disk with maximum radius centered at $g$ and contained in $P$. Denote by $R_g$ the radius of $D_g$. If $g$ $\alpha$-robustly guards $p$, then $D(g,\alpha\|g-p\|)$ is contained in $P$, and thus $\alpha\|g-p\|\le R_g$. We obtain the following observation.
\begin{observation}\label{obs:p_is_close}
    If $g$ $\alpha$-robustly guards $p$, then $\|g-p\|\le R_g/\alpha$.
\end{observation}

\begin{figure}[htb]
	\centering
	\includegraphics[scale=0.55]{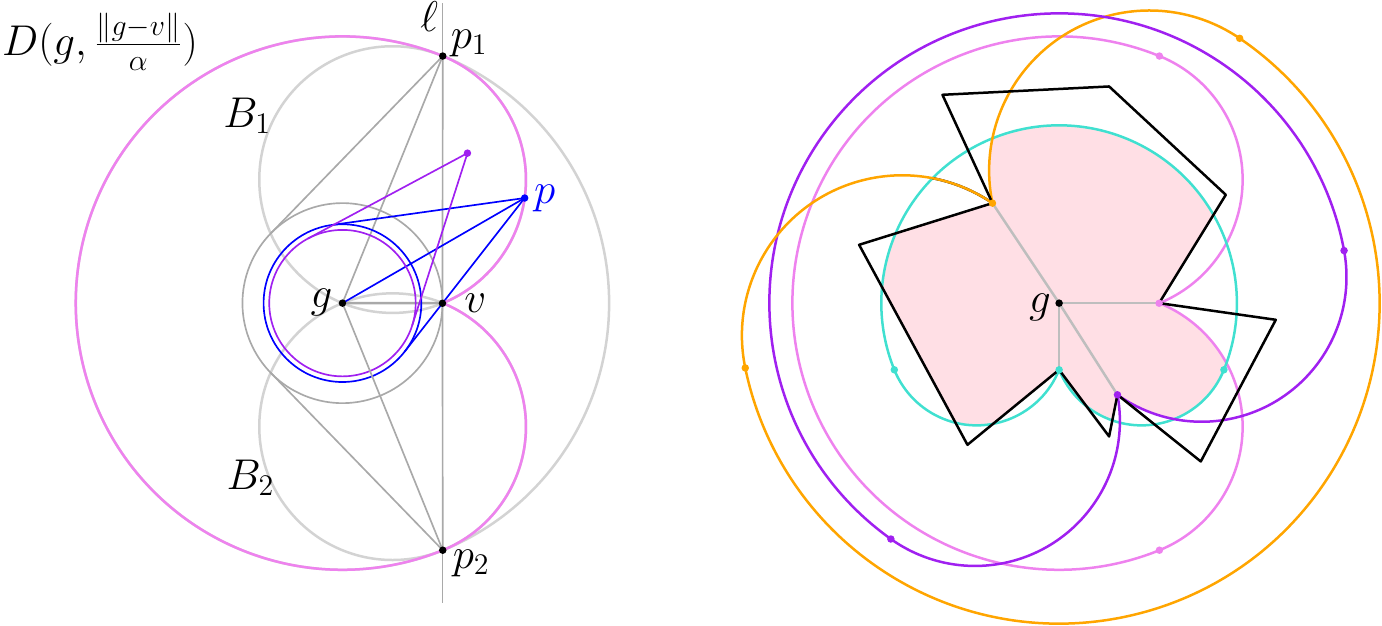}
 \vspace{-0.2cm}
	\caption{\label{fig:computing_vis_1_0} Computing $\VP_\alpha(g)$ as the intersection of heart shapes for every reflex vertex of $\VP(g)$. Left: the construction of a single heart shape (in violet). Right: $\VP_\alpha(g)$ is the area shaded in pink.}
\end{figure}

By \Cref{obs:p_is_close}, we know that $\VP_\alpha(g)$ is contained in $D(g,R_g/\alpha)$. Moreover, by definition, $\VP_\alpha(g)\subseteq \VP(g)$. A point $p\in \VP(g)$ belongs to $\VP_\alpha(g)$ if and only if the ice cream cone from $p$ to $g$ is contained in $\VP_\alpha(g)$. Hence, to compute $\VP_\alpha(g)$ exactly, we need to take into account each of the reflex vertices that may ``block'' a potential ice cream cone. We thus compute for each reflex vertex $v$ the locus of all points $p$ such that the ice cream cone from $p$ to $g$ is touching $v$. The ice cream cone may touch $v$ either at a point on its circular arc or on one of its edges, so we get three relevant circles for each vertex $v$ (see \Cref{fig:computing_vis_1_0}, left). We then compute $\VP_\alpha(g)$ as the intersection of those heart-shaped regions (one for each reflex vertex), $D(g,R_g/\alpha)$, and the boundary of $P$ (see \Cref{fig:computing_vis_1_0}, right). A full proof of \Cref{clm:computing_vis_1_0} is given in \Cref{sec:proofs_robust_vis}.

\begin{restatable}{lemma}{computingVis}\label{clm:computing_vis_1_0}
    Computing $\VP_\alpha(g)$ can be done in polynomial time.     
\end{restatable}

\subsection{The robust inverse visibility region}\label{sec:inv_vis}
The definition of $\alpha$-robust guarding is not bidirectional; it is possible that $g$ $\alpha$-robustly guards $p$, but $g$ is not $\alpha$-robustly guarded by $p$. (Later, we discuss a more general notion of visibility (\Cref{def:robust_visibility_general}) that is bidirectional.) We therefore define the \dfn{robust inverse visibility region} of a point as follows. For a point $p\in P$, denote by $\VP^{\text{inv}}_\alpha(p)$ the set of points $g\in P$ such that $p$ is $\alpha$-robustly guarded by $g$. Although $\VP_\alpha(g)$ is fat (as shown in \Cref{clm:VP_1_0_fat}), $\VP^{\text{inv}}_\alpha(p)$ is not necessarily fat; in fact, the robust inverse visibility region may be a single line segment (see \Cref{fig:line_visibility}). 

\begin{figure}[h]
	\centering
	\includegraphics[scale=1]{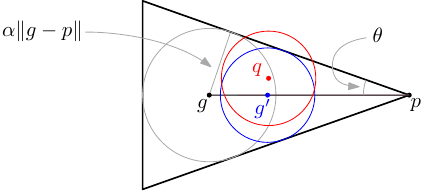}
	\caption{\label{fig:line_visibility} An isosceles triangle polygon with angle exactly $2\theta$ at the point $p$. For any point $g'$ on the segment $\overline{gp}$, the (blue) disk $D(g',\alpha\|p-g'\|)$ is tangent to the legs of the triangle, and thus $g'$ $\alpha$-robustly guards $p$. However, any point not on the segment does not $\alpha$-robustly guard $p$. Thus $\VP^{\text{inv}}_\alpha(p)=\overline{gp}$.}
\end{figure}

Nevertheless, we will show how to construct a star-shaped $O(\alpha)$-fat polygon $F_p$ that contains $\VP^{\text{inv}}_\alpha(p)$ with the property that any $g\in F_p$ $\alpha/2$-robustly guards $p$. First we need the following claim, which is illustrated in \Cref{fig:inv_fat_kite}. The proof is deferred to \Cref{sec:proofs_inv_vis}.

\begin{figure}[h]
	\centering
	\includegraphics[scale=1]{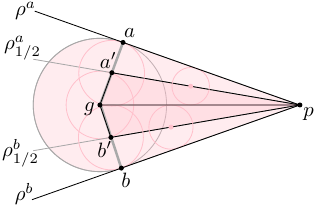}
	\caption{\label{fig:inv_fat_kite} The kite $K=ga'pb'$ is $O(\alpha)$-fat, and any $q\in K$ $\alpha/2$-robustly guards $p$.}
\end{figure}

\begin{restatable}{claim}{VPFatKite}\label{clm:VP_0_1_fat_kite}
Let $g$ and $p$ be points in $P$, such that $g$ $\alpha$-robustly guards $p$. Then there exists an $O(\alpha)$-fat kite $K$ containing $g$ and $p$, such that every point $q\in K$ $\alpha/2$-robustly guards $p$.
\end{restatable}

The lemma below now follows by taking the union of kites corresponding to every $g\in \VP^{\text{inv}}_\alpha(p)$. The complete proof is deferred to  \Cref{sec:proofs_inv_vis}.

\begin{restatable}{lemma}{approxAlpha}\label{lem:approx-alpha}
    Given a polygon $P$ and a point $p\in P$, there exists a star-shaped $O(\alpha)$-fat polygon $F_p$ that contains $\VP^{\text{inv}}_\alpha(p)$, and such that any $g\in F_p$ $\alpha/2$-robustly guards $p$. The size (radius of the smallest enclosing disk centered at $p$) of $F_p$ is equal to that of $\VP^{\text{inv}}_\alpha(p)$.
\end{restatable}

By definition, $\VP^{\text{inv}}_\alpha(p)\subseteq \VP(p)$. Computing $\VP^{\text{inv}}_\alpha(p)$ can also be done in polynomial time, by computing a constant number of ``constraints'' per edge of $\VP(p)$ (see \Cref{fig:computing_vis_0_1}). A point $g\in P$ belongs to $\VP^{\text{inv}}_\alpha(p)$ if and only if the disk $D(g,\alpha\|p-g\|)$ is contained in $\VP(p)$, or in other words, does not intersect any edge of $\VP(p)$. For each edge $e=\{u,v\}\in P$, the locus of all points $g$ such that $D(g,\alpha\|p-g\|)$ touches $e$, can be described by two disks (one per vertex) and a hyperbola (for the interior of the edge). A full proof of \Cref{clm:computing_vis_0_1} is given in \Cref{sec:proofs_inv_vis}.

\begin{restatable}{lemma}{computingVisInv}\label{clm:computing_vis_0_1}
    The region $\VP^{\text{inv}}_\alpha(p)$ can be computed in polynomial time.
\end{restatable}

\begin{figure}[htb]
	\centering
	\includegraphics[scale=0.8]{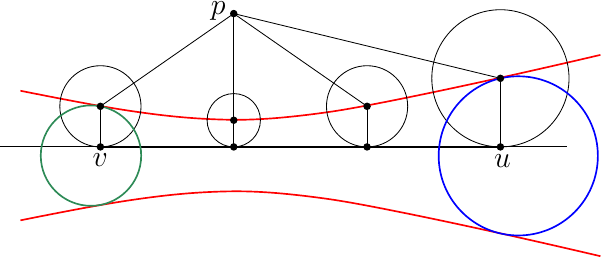}
	\caption{\label{fig:computing_vis_0_1} Three constraints defining the points $g$ with $D(g,\alpha\|p-g\|)$ intersecting the edge $\{u,v\}$: the green disk containing $v$, the blue disk containing $u$, and the red hyperbola.}
\end{figure}

\subsection{Hardness}
The classic Art Gallery Problem is APX-hard by a reduction from the {\sc Hitting Lines} problem~\cite{broden2001guarding}: Given a set ${\cal L}$ of lines, one is to find a minimum set of points that ``hit'' all the lines. The polygon constructed in the reduction is a ``spike box'' -- a rectangle containing all the intersection points between lines in ${\cal L}$, and having a thin spike going out of it for each line. In order to guard the tip of a spike, one must place a guard in a small neighbourhood of the line segment (corresponding to a line in ${\cal L}$) generating the spike. Thus, hitting all of the lines is equivalent to guarding all spikes. In \Cref{sec:proof_hardness} we show a similar construction for the problem of $\alpha$-robust guarding, and thereby obtain the following theorem.

\begin{restatable}{theorem}{hardness}\label{thm:hardness}
The $\alpha$-robust guarding problem is APX-hard.
\end{restatable}

\subsection{Robustly guarding a discrete set of points}
In this section we consider a discrete version of the robust guarding problem, where we are given a set $S$ of $m$ points in a polygon $P$, and the goal is to find a minimum set of $\alpha$-robust guards for $S$. Besides being interesting in its own, the solution that we present will be used in the next section, where the goal is to $\alpha$-robustly guard the entire polygon $P$.

Before we can present our algorithm, we need one more ingredient, regarding the fatness of our robust visibility polygons. The lemma below is a version of the well-known ``fat-collection theorem'' (see, e.g.~\cite{StappenO94}). Here, we define the \dfn{size} of a star-shaped object $P$ with respect to a given center point $o$ as the radius of its minimum enclosing ball centered at $o$. For completeness, we provide a proof in \Cref{sec:proof-fat-collection}.

\begin{restatable}{lemma}{fatCollection}\label{lem:fat-collection}
    For any disk $D$ of radius $R$, there exist a set $C$ of $O(\alpha^{-2})$ points such that any $\alpha$-fat star-shaped polygon that intersects $D$ and has size at least $R$ w.r.t. a given center point $o$, contains a point from $C$.    
\end{restatable}

\oldparagraph{The algorithm.} Consider \Cref{alg:robust-greedy-approx}, which gets as input the polygon $P$ and the set $S$. In each iteration, the algorithm finds the point $g\in S$ with smallest $\VP^{\text{inv}}_\alpha(g)$, removes from $S$ all the points $s$ for which $\VP^{\text{inv}}_\alpha(s)$ intersects $\VP^{\text{inv}}_\alpha(g)$, and adds to the solution the corresponding set of hitting points from \Cref{lem:fat-collection}.

\begin{algorithm}[htb]\caption{\textsc{DiscreteRobustGuarding}($P,S$)}\label{alg:robust-greedy-approx}
    \ForEach{$s\in S$}{
        Compute $\VP^{\text{inv}}_\alpha(s)$\\
        Compute $D(s)$, the minimum enclosing disk of $\VP^{\text{inv}}_\alpha(s)$ centered at $s$
    }
    $G\gets\emptyset$\\
    \While{$S\neq\emptyset$}{
	$g \gets \argmin_{s\in S}\{\mathrm{size}(\VP^{\text{inv}}_\alpha(s))\}$  \hfil \CommentSty{$g$ is the point from $S$ with smallest $\VP^{\text{inv}}_\alpha$} \\
        $S(g)\gets \{s\in S\mid\VP^{\text{inv}}_\alpha(s)\cap D(g)\neq\emptyset\}$\\
        $S\gets S\setminus S(g)$\\
	Let $H(g)$ be the set of hitting points from \Cref{lem:fat-collection} that correspond to $D(g)$, with fatness parameter $c\cdot\alpha$ (for a sufficiently small constant $c$).\\
        $G\gets G\cup H(g)$
    }
    \Return{$G$}
\end{algorithm}

\begin{theorem}\label{thm:robust-greedy-approx-discrete}
    Given a polygon $P$ with $n$ vertices, and a set $S$ of $m$ points in $P$, one can compute in $\poly(n,m)$ time a set of $O(\alpha^{-2})|\OPT_\alpha^S|$ points that $\alpha/2$-robustly guard $S$, where $\OPT_\alpha^S$ is a minimum set of guards that $\alpha$-robustly guard $S$.
\end{theorem}
\begin{proof}
    We show that the set $G$ returned by \Cref{alg:robust-greedy-approx} satisfies the theorem.
    
    Let $g_1,\dots,g_k$ be the points from $S$ that were found in line 6 of the algorithm, and consider the sequence of inverse visibility regions $\VP^{\text{inv}}_\alpha(g_1),\dots,\VP^{\text{inv}}_\alpha(g_k)$. Any two of these regions are disjoint, because in line 8 we remove all points $s\in S$ for which $\VP^{\text{inv}}_\alpha(s) \cap \VP^{\text{inv}}_\alpha(g_i) \ne \emptyset$. Thus, the set $g_1,\dots,g_k$ is a set of witnesses in $P$, i.e., no $\alpha$-robust guard can guard both $g_i$ and $g_j$ for any $1\le i,j\le k$. Therefore, in order to $\alpha$-robustly guard $S$, one needs to put a guard in each $\VP^{\text{inv}}_\alpha(g_i)$, and hence $k\le |\OPT_\alpha^S|$.
	
    The set $H(g)$ is the set of $O(\alpha^{-2})$ hitting points obtained from \Cref{lem:fat-collection}. Since we chose $g_i$ to be the point in $S$ with minimum size (i.e., size of inverse visibility region), any $s\in S(g_i)$ has size larger than the radius of $D(g_i)$. By \Cref{lem:approx-alpha}, for any $s\in S(g_i)$ there exists a star-shaped $O(\alpha)$-fat polygon $F_s$ of size $\mathrm{size}(\VP^{\text{inv}}_\alpha(s))$ that contains $\VP^{\text{inv}}_\alpha(s)$, and any point in $F_s$, $\alpha/2$-robustly guards $s$. Thus by \Cref{lem:fat-collection} each such $F_s$ contains a point $g_s$ from $H(g_i)$, and $g_s$ $\alpha/2$-robustly guards $s$. We get that for any $1\le i\le k$, the set $H(g_i)$ guards $S(g_i)$, and therefore $G$ is a set of $O(\alpha^{-2})|\OPT_\alpha^S|$ points that $\alpha/2$-robustly guards $S$. (Note that we do not compute $F_s$, we use \Cref{lem:approx-alpha} only to show that the set of hitting points is sufficient.)

    By \Cref{clm:computing_vis_0_1}, computing $\VP_\alpha^{-1}(s)$ for every $s\in S$ can be done in $m\cdot\poly(n)$ time. Clearly, the while loop is executed for at most $m$ rounds, each round runs in $\poly(n,m)$ time. Thus the total running time of \Cref{alg:robust-greedy-approx} is in $\poly(n,m)$.
\end{proof}

%% file: o-sec_const_approx.tex
\section{An $O(1)$-approximation for robustly guarding a polygon}\label{sec:robust-guarding-algorithm}
Let $P$ be a polygon with $n$ vertices, and a parameter $\alpha\le \sin\frac{\phi}{2}$, where $\phi$ is the smallest internal angle of $P$. For technical reasons, we also assume that $\alpha\le 1/2$. 
Let $\OPT_\alpha$ be a minimum set of points that $\alpha$-robustly guard $P$.
Our goal is to find a set of $O(\poly(\alpha^{-1}))\cdot|\OPT_\alpha|$ points that $c\alpha$-robustly guard $P$, for some constant $c\le1$. Note that for a smaller radius we need less guards, i.e., for $\alpha'<\alpha$, $|\OPT_{\alpha'}|\le|\OPT_\alpha|$.

\subsection{A medial axis based decomposition}\label{sec:medial-axis}
Consider the medial axis of $P$ (the set of all points in $P$ having more than one closest point on $\partial P$ -- the boundary of $P$), and let $M$ be the set of its vertices that do not lie on $\partial P$. The medial axis is a planar graph $G$, with some line-segment edges and some curved edges (subcurves of a parabola). Note that we also include in $M$ vertices of degree 2 that represent, e.g., the intersection point of a line segment and a parabola that define the medial axis (see \Cref{fig:medial_axis_cells}).
For each $v\in M$, let $D_v$ be the \emph{medial} disk centered at $v$. The disk $D_v$ touches $\partial P$ in at least two points, and we denote by $\D$ the set of all medial disks, i.e., $\D=\{D_v\}_{v\in M}$.

\begin{figure}[h]
	\centering
	\includegraphics{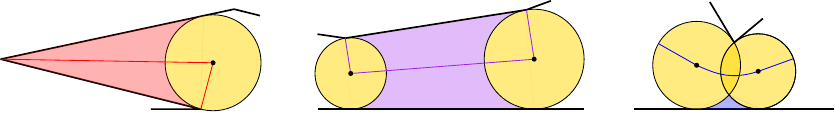}
	\caption{The three types of regions (cells) forming $P\setminus \D$; a red cell (left), a purple cell (middle), and a blue cell (right). The medial disks in $\D$ are shown in yellow.}
	\label{fig:medial_axis_cells}
\end{figure}

Based on the structure of the medial axis, we decompose $P\setminus \D$ into 3 types of regions, Red, Purple, and Blue, as follows (see \Cref{fig:medial_axis_cells}).
\begin{itemize}
\itemsep0em 
	\item A Red region is a maximal connected region of $P\setminus \D$ which is bounded by two edges $e_1,e_2$ of $P$ connected by a vertex $w$, and a single disk $D_v\in\D$, such that both $e_1$ and $e_2$ are tangent to $D_v$.
	\item A Purple region is a maximal connected region of $P\setminus \D$ which is bounded by two disjoint edges $e_1,e_2$ of $P$, and two disks $D_v,D_w\in\D$, such that each of $e_1,e_2$ is tangent to both $D_v,D_w$. Since there is no other feature of the polygon ``between'' $D_v$ and $D_w$, the corresponding vertices $v,w\in M$ are connected by a line segment in the medial axis.
	\item A Blue region is any maximal connected region of $P\setminus \D$ which is neither red nor purple.
\end{itemize}

\oldparagraph{Adding medial disks in purple regions.} In each purple region we add to $M$ a set of vertices as follows.
Let $v,w$ be the two medial vertices that define a purple region $\Pi$, and assume that $R_v\ge R_w$ (see \Cref{fig:purple_regions}). 
Consider the intersection $I=D(v,R_v/\alpha)\cap \Pi$, and notice that any point in $I$ is $\alpha$-robustly visible from $v$.
If $D(v,R_v/\alpha)$ does not contain $\Pi$, then there are two intersection points, $q_1,q_2$, between $D(v,R_v/\alpha)$ and the edges defining $\Pi$. Let $p_1$ be the point on the segment $\overline{vw}$ such that the medial disk centered at $p_1$ touches the edges defining $\Pi$ at the points $q_1,q_2$. We add $p_1$ to $M$, set $v=p_1$ and repeat the process, i.e. while $D(p_i,R_{p_i}/\alpha)$ does not contain the part of $\Pi$ between $D(p_i,R_{p_i})$ and $D_w$, add to $M$ the point $p_{i+1}$ on the segment $\overline{vw}$, defined similarly but with respect to the disk $D(p_i,R_{p_i}/\alpha)$. Note that by adding the sequence $p_1,\dots,p_k$ of additional medial vertices to $M$, we subdivide the purple region into $k+1$ smaller purple regions.

\begin{figure}[h]
	\centering
	\includegraphics[scale=1.1]{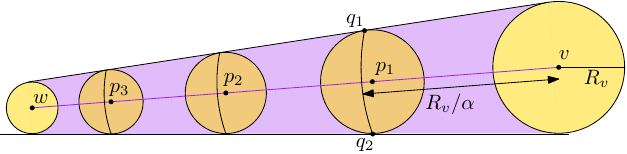}
	\caption{A purple region between $D_w$ and $D_v$, and the added sequence of medial disks $p_1,p_2,p_3$.}
	\label{fig:purple_regions}
\end{figure}

Intuitively, those disks in the interior of a purple region $\Pi$ were added in such a way that a single guard does not see too many of them robustly. More precisely, we have the following observation, which we prove in \Cref{sec:proofs-const-approx}.
\begin{restatable}{observation}{purpleRegions}\label{obs:purple_regions}
    For any $g\in P$, $\VP_\alpha(g)$ intersects at most four of the disks $D_{p_i}$ that were added to $\Pi$.
\end{restatable}

The observation below follows from the definitions of Purple and Red regions, and the observation that any blue region is bounded by two non-disjoint medial disks and one edge of $P$. A formal proof is given in \Cref{sec:proofs-const-approx}.
\begin{restatable}{observation}{boundaryCells}\label{obs:boundary-cells}
    Any Red, Blue, or Purple region has at most two disks defining its boundary. 
\end{restatable}

\oldparagraph{Associating points with at most two medial disks.} Given some point $p\in P$, we associate $p$ with either one or two medial disks from $\D$ as follows. If $p$ is contained in some disk from $\D$, we associate $p$ with the largest disk from $\D$ that contains it. Otherwise, $p$ belongs to either a Red, Blue, or Purple region, and we associate $p$ with the disks from $\D$ defining that cell. (by \Cref{obs:boundary-cells}, there are at most two such disks).

\begin{observation}\label{obs:p_visible_by_medial_vertex}
Let $p$ be a point in $P$.
    \begin{enumerate}[(i)]
    \item If $p\in D_v$ for some medial vertex $v$, then $p$ is $\alpha$-robustly visible from $v$.
    \item If $p$ is in a Red, Blue, or Purple region, then $p$ is $\alpha$-robustly visible from one of the centers of medial disks associated with $p$.
    \end{enumerate}
\end{observation}
\begin{proof}
    The first statement is trivial. 
    For the second statement, if $p$ is in a Red region bounded by a medial disk $D_v$, then clearly $p$ is visible from $v$ (recall that we assume that convex angles in $P$ are larger than $2\theta$).
    Else, if $p$ is in a Blue region bounded by two non-disjoint medial disks $D_v$ and $D_u$, then since $\alpha\le\frac12$ one of the following holds: (i) $\|p-v\|\le R_v/\alpha$, and then $p$ is $\alpha$-robustly visible from $v$,
    or (ii) $\|p-u\|\le R_u/\alpha$, and then $p$ is $\alpha$-robustly visible from $u$. 
    Else, $p$ is in a Purple region bounded by two medial disks $D_v$ and $D_u$, such that $R_v\ge R_u$.
    By the construction of additional disks in purple regions, $v$ sees the entire purple region $\alpha$-robustly.
\end{proof}

\begin{observation}\label{obs:largest_disk}
Let $g,p$ be two points in $P$ such that $g$ $\alpha$-robustly guards $p$, and let $D_v$ be the largest disk associated with $g$. Then $\alpha\|g-p\|\le R_v$.
\end{observation}
\begin{proof}
For an edge $\{u,v\}$ of the medial axis with $R_u\le R_v$, the radii of maximal disks with center on the edge is at most $R_v$. Points not on a medial edge clearly have smaller radius, and thus $R_g\le R_v$ (recall that $R_g$ is the radius of the largest disk centered at $g$ and contained in $P$). By \Cref{obs:p_is_close}, 
we get $\alpha\|p-g\|\le R_g\le R_v$.
\end{proof}

\subsection{A set of candidate guards}\label{sec:robust-candidate-guards}
For each $v\in M$ we place on $D_v$ a set $Q_v$ of $\Theta(\alpha^{-4})$ grid points, with edge length in $\Theta(\alpha^2 R_v)$. Denote $Q=M\cup\bigcup_{v\in M}Q_v$. For a point $g\in P$, let $Q(g)=\left\{Q_v\cup\{v\}\mid D_v\text{ is associated with } g\right\}$. 
As there are at most two disks associated with $g$, we have $|Q(g)|=O(\alpha^{-4})$. In this subsection we show that any $\alpha$-robust guard $g$ can be replaced by the set $Q(g)$ of $\alpha/4$-robust guards.

\begin{figure}[ht]
    \centering
    \includegraphics[scale=1.2]{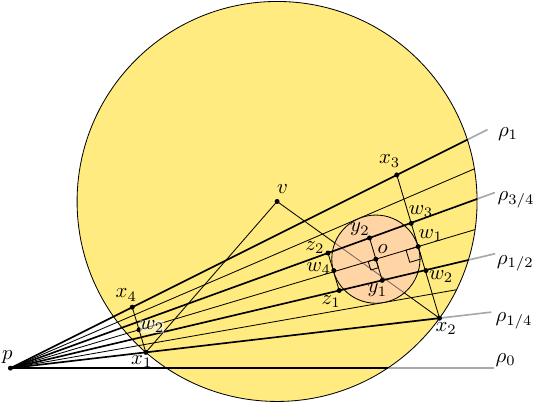}
    \caption{The construction for the proof of \Cref{lem:cone_intersecting_disk}.}     \label{fig:cone_intersecting_disk}
\end{figure}
\begin{lemma}\label{lem:cone_intersecting_disk}
    Let $K$ be a cone defined by two rays $\rho_0,\rho_1$ originated at $p$ with small angle $\theta$. If both $\rho_0,\rho_1$ intersect $D_v$, and $p$ sees $K\cap D_v$, then there exists a grid point in $Q_v$ that $\frac\alpha4$-robustly guards $p$.
\end{lemma}
\begin{proof}
    Assume for simplicity that $\rho_0$ lies on the $x$-axis, and $\rho_1$ lies above it. For $0<\gamma<1$, denote by $\rho_{\gamma}$ the ray from $p$ between $\rho_0$ and $\rho_1$ with angle $\gamma\cdot\theta$ from $\rho_0$.
    First, by \Cref{obs:disk_in_cone} for any point $q$ that lies in $K$ between $\rho_{1/4}$ and $\rho_{3/4}$, the disk $D(q,\frac{\alpha}{4}|p-q|)$ is contained in $K$. In addition, if $v$ lies in $K$, then since both $\rho_0,\rho_1$ intersect $D_v$, we get that $R_v$ is at least the distance between $v$ and one of $\rho_0,\rho_1$. As in the proof of \Cref{obs:disk_in_cone}, we get that $R_v\ge \sin\frac\theta2\|p-v\|>\frac{\sin\theta}{2}\|p-v\|=\frac\alpha2\|p-v\|$. Therefore, if $v$ lies between $\rho_{1/4}$ and $\rho_{3/4}$, then clearly $D(v,\frac{\alpha}{4}\|p-v\|)$ is contained in $K\cap D_v$. 
    
    We thus assume w.l.o.g. that $v$ lies above $\rho_{3/4}$ (the case when $v$ lies below $\rho_{1/4}$ is symmetric).
    Our goal is to find a large enough square that lies in $K\cap D_v$ between $\rho_{1/4}$ and $\rho_{3/4}$, and such that for any point $q$ in that square, $D(q,\frac{\alpha}{4}|p-q|)$ is contained in $D_v$. 
    
    In the following, we use some trigonometric identities, and the Maclaurin series expansions of some trigonometric functions (for $x\le \frac12$) to estimate distances. Specifically, $\tan(x)=\frac{\sin(x)}{\cos(x)}$, $\sin(2x)=2\sin(x)\cos(x)$, $x\ge\sin(x)\ge x-\frac{x^3}{3!}$, $1-\frac{x^2}{2}+\frac{x^4}{4!}\ge \cos(x)\ge 1-\frac{x^2}{2}$, $2x\ge \arcsin(x)\ge x$, $2x\ge \tan(x)\ge x$.

    Let $x_1,x_2$ be the two intersection points of $\rho_{1/4}$ and $D_v$ (see \Cref{fig:cone_intersecting_disk}).
    We have 
    $$\|x_1-x_2\|\ge 2\sin(\theta/4)\cdot R_v= \Theta(\theta)\cdot R_v$$ (we will get an equality when $p$ is on $\partial D_v$, and $\rho_0$ is tangent to $D_v$). Consider the ray from $x_2$ perpendicular to $\rho_{5/8}$, and let $w_1$ be the intersection point with $\rho_{5/8}$, and $x_3$ the intersection point with $\rho_1$. Similarly, consider the ray from $x_1$ perpendicular to $\rho_{5/8}$, and let $w_2$ be the intersection point with $\rho_{5/8}$, and $x_4$ the intersection point with $\rho_1$. Since $v$ is above $\rho_{1/2}$, the quadrilateral $x_1x_2x_3x_4$ is contained in $D_v$. Note that if $p$ is on $\partial P$, we get $p=x_1=x_4$, and $\triangle px_2x_3$ is a triangle contained in $D_v$.
    
    We have 
    $$\|p-w_1\|= \cos(\frac{3\theta}{8})\|p-x_2\|\ge \cos(\frac{3\theta}{8})\|x_1-x_2\|\ge (1-\frac{9\theta^2}{64})\|x_1-x_2\| =\Theta(\theta)\cdot R_v,$$
    and 
    $$\|x_2-w_1\|\ge \tan(\frac{3\theta}{8})\|p-w_1\|\ge \sin(\frac{3\theta}{8})\|p-x_2\|\ge \sin(\frac{3\theta}{8})\|x_1-x_2\|=\Theta(\theta^2)\cdot R_v.$$ 
    
    Let $o$ be the point on $\overline{pw_1}$ such that $\|w_1-o\|=\alpha/4\|o-p\|$, then the disk $D(o,\alpha/4\|o-p\|)$ is contained in $K\cap D_v$. 
    Let $y_1,y_2$ be the points on $\rho_{1/2},\rho_{3/4}$, respectively, such that the line through $y_1,y_2$ is the perpendicular to $\overline{pw_1}$ at $o$. Let $z_1,z_2$ be the points on $\rho_{1/2},\rho_{3/4}$, respectively, such that the line through $z_1,z_2$ is the perpendicular to $\overline{pw_1}$ at the other intersection point, $w_4$, of $D(o,\alpha/4\|o-p\|)$ and $\overline{pw_1}$. For any point $q$ in the quadrilateral $z_1,y_1,y_2,z_2$, the disk $D(q,\alpha/4\|p-q\|)$ is contained in $D_v$. 
    We now claim that the quadrilateral $z_1y_1y_2 z_2$ contains a disk of diameter in $\Theta(\|x_2-w_1\|)=\Theta(\alpha^2)\cdot R_v$, and thus it must contain a grid point from $Q_v$. Intuitively, this is true because $\|z_1-z_2\|=\Theta(\|x_2-w_1\|)$ and $\|o-w_4\|=\alpha/4\cdot\Theta(\|p-w_1\|)$. We provide detailed calculations for this claim in \Cref{sec:proofs_candidates} (\Cref{claim:quad}), which finishes the proof.
\end{proof}

\begin{lemma}\label{lem:robust-vision-candidates}
	For any $g\in P$, $\VP_\alpha(g)\subseteq\cup_{q\in Q(g)}\VP_{\alpha/4}(q)$.
\end{lemma}
\begin{proof}
Let $g,p$ be two points in $P$, such that $p$ is $\alpha$-robustly guarded by $g$. 
By \Cref{obs:p_visible_by_medial_vertex}, if $p$ is in one of the disks associated with $g$, or both $p$ and $g$ are in a Red, Blue, or Purple region, then $p$ is $\alpha$-robustly guarded by the centers of the disks associated with $g$. If this is not the case, then $p,g$ must be in different cells of the arrangement $\A$, and therefore the segment $\overline{pg}$ must cross the boundary of some disk $D_v$ associated with $g$.

Let $\rho_0$ be the ray from $p$ in the direction of $g$, and assume for simplicity that $\rho_0$ lies on the $x$-axis.
Denote by $\rho^a$ (resp. $\rho^b$) the ray from $p$ tangent to $D(g,\alpha\|p-g\|)$ above (resp. below) the $x$-axis. 
Denote by $\rho^a_{\gamma}$ (resp. $\rho^b_{\gamma}$) the ray from $p$ between $\rho_0$ and $\rho^a$ (resp. $\rho^b$) with angle $\gamma\cdot\theta$ from $\rho_0$.
Let $a$ (resp. $b$) be the intersection point of $\rho^a$ (resp. $\rho^b$) and $D(g,\alpha\|p-g\|)$. 

Assume w.l.o.g. that $v$ is above $\rho_0$, and let $K$ be the cone defined by $\rho_0$ and $\rho^a$. Recall that $D_v$ do not contain $p$, and $D_v\cap\overline{pg}\neq\emptyset$.

If $D_v\cap\overline{pa}\neq\emptyset$, then there exists a point $q_1\in D_v\cap\overline{pa}$ and a point $q_2\in D_v\cap\overline{pg}$ such that the triangle $\triangle pq_1q_2$ is contained in $P$, and since $q_1\in\rho^a$ and $q_2\in\rho_0$ we get that $p$ sees $K\cap D_v$. Hence we can apply \Cref{lem:cone_intersecting_disk} on the rays $\rho_0,\rho^a$ and get that there exists a grid point in $Q_v$ that $\alpha/4$-robustly guards $p$.

Else, we are in the case when $D_v\cap\overline{pa}=\emptyset$. First, we claim that $g\in D_v$. Indeed, if $g\notin D_v$, then since $D_v\cap\overline{pg}\neq\emptyset$ it must be that $v$ is between the vertical line trough $g$ and the vertical line through $p$. Now, if $v$ is above $\rho^a$, then clearly $D_v\cap\overline{pa}\neq\emptyset$. Else, if $v$ is below $\rho^a$, then since by \Cref{obs:largest_disk} we have $R_v\ge\alpha\|p-g\|$, again we get that $D_v\cap\overline{pa}\neq\emptyset$.

Therefore, we are left with the following scenario: $v$ is above $\rho_0$, $g\in D_v$, $a\notin D_v$, and $R_v\ge\alpha\|p-g\|$.
Denote by $w_1,w_2$ the intersection points of $D_v$ and $D(g,\alpha\|p-g\|)$, and consider the rays $\rho_{w_1},\rho_{w_2}$ from $p$ to $w_1,w_2$, respectively. Since $p$ sees both $w_1,w_2$, we get that $p$ sees the cone between $\rho_{w_1}$ and $\rho_{w_2}$.
In ~\Cref{sec:proofs_candidates}, we prove that $\angle w_1pw_2\ge\theta$ (\Cref{clm:angle_is_large}). This shows that we can apply \Cref{lem:cone_intersecting_disk} on this cone, and find a grid point in $Q_v$ that $\alpha/4$-robustly sees $p$.
\end{proof}

\oldparagraph{}
By replacing each $g\in \OPT_\alpha$ with the set $Q(g)$, we get that the set $\bigcup_{g\in\OPT}Q(g)$ $\alpha/4$-robustly guards $P$. We obtain the following theorem.
\begin{theorem}\label{thm:robust-vision-candidates}
The set $Q=M\cup\bigcup_{v\in M}Q_v$ contains a set of $O(\alpha^{-4})|OPT_\alpha|$ points that $\alpha/4$-robustly guard $P$.
\end{theorem}

In addition, we claim that the size of $Q$ is linear in $n=|P|$ and $|OPT_\alpha|$.
\begin{claim}\label{clm:Q_size}
    $|Q|=O(\alpha^{-4})(n+|OPT_\alpha|)$.
\end{claim}
\begin{proof}
    It is well known that the number of vertices that define the medial axis is $O(n)$. We only need to show that the number of vertices that we add in the purple regions is $O(|OPT_\alpha|)$. Indeed, by \Cref{obs:purple_regions}, for any $g\in P$ $\VP_\alpha(g)$ intersects at most four such consecutive disks, and thus the number of guards from $\OPT_\alpha$ in a purple region with $k$ additional disks is at least $\frac{k-8}{4}$ ($4$ from each side can be guarded by a point outside of the purple region).
\end{proof}

\subsection{An $O(1)$-approximation greedy algorithm}
Let $Q$ be the set of candidate guards from \Cref{thm:robust-vision-candidates}, constructed with parameter $\alpha/8$ instead of $\alpha$.
Consider the arrangement $\A$ formed by the set of visibility regions  $\{\VP_{\alpha/8}(q) \mid q\in Q\}$. For each cell of this arrangement, we pick one sample point in the interior of the cell, and denote by $S$ the set of these sample points.

\begin{observation}\label{obs:robust-sample-points}
	If $Q'\subseteq Q$ $\frac\alpha8$-robustly guards $S$, then $Q'$ $\frac\alpha8$-robustly guards all of $P$.
\end{observation}
\begin{proof}
	Any guard that $\alpha/8$-robustly guards a point in the interior of a cell in $\A$ must $\alpha/8$-robustly guard the entire cell; otherwise, this cell would be subdivided.
\end{proof}

By \Cref{obs:robust-sample-points}, it is enough to $\alpha/8$-robustly guard $S$ from points in $Q$. We run \Cref{alg:robust-greedy-approx} on $P$ and the set $S$, and by \Cref{thm:robust-greedy-approx-discrete} we get a set $G$ of $O(\alpha^{-2})|\OPT^S_\alpha|\le O(\alpha^{-2})|\OPT_\alpha|$ points that $\alpha/2$-robustly guard $S$. However, in order to guard the entire polygon $P$, we need to guard $S$ from points in $Q$ only.
So, we replace each point $g\in G$ by the set $Q(g)$ from \Cref{lem:robust-vision-candidates}. For any $g\in P$ we have $|Q(g)|=O(\alpha^{-4})$, so we obtain a set $Q'$ of $O(\alpha^{-6})|OPT_\alpha|$ points from $Q$ that $\alpha/8$-robustly guard $S$. By \Cref{obs:robust-sample-points}, the set $Q'$ $\alpha/8$-robustly guards $P$.

Computing the set $Q$, the arrangement $\A$, and the set $S$ can be done in $\poly(n,|\OPT_\alpha|)$ time by \Cref{clm:computing_vis_1_0} and \Cref{clm:Q_size}. By \Cref{thm:robust-greedy-approx-discrete}, the running time of \Cref{alg:robust-greedy-approx} is $\poly(n,|\OPT_\alpha|)$. 

In fact, if we only want to return a constant factor approximation of $|\OPT_\alpha|$, we can do so in $\poly(n)$ time, as follows. By \Cref{obs:purple_regions}, if the number of additional disks placed in a purple region $\Pi$ is $k>8$, then the number of guards placed in $\Pi$ in an optimal solution is $\Omega(k)$. 
Moreover, except for $O(1)$ of these guards, none of them sees points outside $\Pi$.
Therefore, we do not need to compute these guards explicitly in order to produce the rest of the guards, and we only record their number (which is a simple function of the dimensions of the purple region) in $O(1)$ time. Thus, we can cut the inner part (region inside a purple region excluding 4 disks at each of its ends; see \Cref{fig:inner_purple})
of those purple regions from $P$, and obtain a set of disjoint subpolygons having in total $O(\alpha^{-4})\cdot n$ candidate grid points. By applying the same algorithm separately on each subpolygon and then combining the solutions, we only loose a constant number of guards per purple region. Therefore we can output a constant factor approximation of $|\OPT_\alpha|$ in $\poly(n)$ time. To produce an explicit solution from this implicit representation, we only need to run the algorithm that computes the set of guards in each of the purple regions, in time linear in their number.

\begin{figure}[ht]
    \centering
    \includegraphics[scale=0.85]{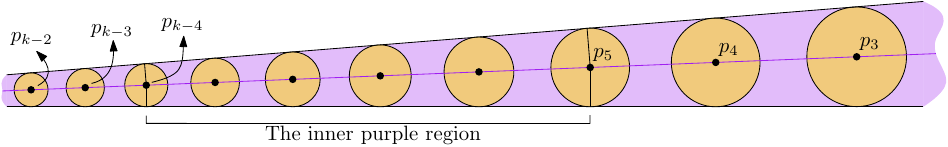}
    \caption{The inner part of a purple region with $k>8$ added medial vertices.}     \label{fig:inner_purple}
\end{figure}

\begin{theorem}\label{thm:robust-greedy-approx}
    Given a polygon $P$ with $n$ vertices, one can compute in $\poly(n)$ time the cardinality of, and an implicit representation of, a set of $O(\alpha^{-6})|\OPT_\alpha|$ points that $\alpha/8$-robustly guard $P$, where $\OPT_\alpha$ is a minimum-cardinality set of guards that $\alpha$-robustly guard $P$. In additional time $O(|\OPT_\alpha|)$ we can output an explicit set of such points.
\end{theorem}

%% file: sec_improved_approx.tex
%\newpage
\section{Extending the definition and reducing the number of guards}\label{sec:smaller-solution}

In \Cref{sec:robust-guarding-algorithm} we have considered a bicriteria approximation for $\alpha$-robust guarding, where we obtain a set $G$ of size $O(\alpha^{-6})\OPT$, such that $P$ is $\alpha/8$-robustly guarded by the points in $G$. 
The dependency on $\alpha$ in our approximation algorithm may seem too large at first. In fact, one may think that using a subset of the medial vertices would result in a good approximation. However, there are some cases in which the restriction to medial axis guards may lead to a linear (in $n$) number of guards. We provide the following example (\Cref{fig:counter2}) to show where such an approach might fail.

\begin{figure}[h]
    \centering
    \includegraphics[page=2,scale=1.8]{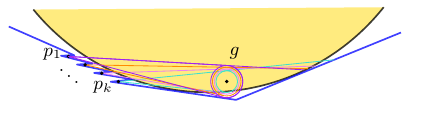}
    \caption{The point $g$ $\alpha$-robustly guards $p_1,\dots,p_k$, and this is the only point that guards all of them robustly. However, $g$ is far from any medial vertex. In addition, any $p_i$ sees a very small portion of the large medial disk (in yellow), which intersects all the ice cream cones.}
    \label{fig:counter2}
\end{figure}

Nonetheless, it is possible to decrease the approximation factor on the number of guards to $O(\alpha^{-3})$, by relaxing the definition of robust guards, such that the disk around the guard is not entirely visible to the points that are robustly guarded by it.

\subsection{A more general definition for robust guarding}
Consider the following definition, which generalizes the definition of $\alpha$-robust guarding.
\begin{definition}\label{def:robust_visibility_general}
Given a polygon $P$ and parameters $0<\beta_{guard}, \beta_{point},\alpha\le 1$, we say that a point $g\in P$ \dfn{$(\beta_{guard}, \beta_{point},\alpha)$-robustly guard} another point $p\in P$ if $\overline{gp}\in P$, and %for any $0<\alpha'\le\alpha$:
\begin{enumerate}
    \item the area of $\VP(p)\cap D(g,\alpha\cdot\|p-g\|)$ is at least $\beta_{guard}\cdot \pi (\alpha\cdot\|p-q\|)^2$, and
    \item the area of $\VP(g)\cap D(p,\alpha\cdot\|p-g\|)$ is at least $\beta_{point}\cdot \pi (\alpha\cdot\|p-g\|)^2$
\end{enumerate}
  In other words, $p$ sees $\beta_{guard}$-fraction of the area of $D(g,\alpha\cdot\|p-g\|)$, and $g$ sees $\beta_{point}$-fraction of the area of $D(p,\alpha\cdot\|p-g\|)$.
\end{definition}

\begin{figure}[h]
	\centering
	\includegraphics{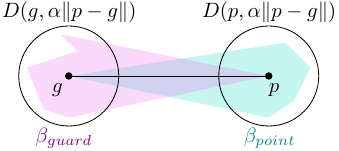}
\end{figure}

The definition of $\alpha$-robust guarding is equivalent to $(1,0,\alpha)$-robust guarding. In addition, we generalize the notation of $\VP_\alpha$ as follows. Denote by $\VP^{\alpha}_{\beta_{guard}, \beta_{point}}(g)$ the \dfn{$(\beta_{guard}, \beta_{point},\alpha)$-robust visibility region} of $g$, i.e. all the points in $P$ that $g$ guards $(\beta_{guard}, \beta_{point},\alpha)$-robustly. 
Here, $\VP_\alpha$ is equivalent to $\VP^{\alpha}_{1,0}$, and $\VP^{-1}_\alpha$ is equivalent to $\VP^{\alpha}_{0,1}$. In this section, we will use these generalized new notations for robust guarding.

\subsection{A set of $O(\alpha^{-2})$ candidate guards for $(1/16,0,\alpha)$-robust guarding}

We now show that it is possible to reduce the approximation factor on the number of robust guards to $O(\alpha^{-3})$, by using $(1/16,0,\alpha/2)$-robust guards instead of $(1,0,\alpha/8)$-robust guards. The idea is to use a smaller set of candidate guards, and then follow the same steps of the greedy algorithm as in \Cref{sec:robust-guarding-algorithm}.

Consider the same medial axis based decomposition as in \Cref{sec:medial-axis}.
For each $v\in M$ we place on the boundary of $D_v$ a set $\hat{Q}_v$ of $\frac{c}{\alpha}$ evenly spaced candidate points. The distance between two consecutive candidate points on $\partial D_v$ is $2\pi R_v\cdot\frac{\alpha}{c}$. Denote $\hat{Q}=M\cup\bigcup_{v\in M}\hat{Q}_v$.
For a point $g\in P$, denote by $\hat{Q}(g)$ the set of candidate points (i.e., grid points) that belong to the disks associated with $g$.

\begin{lemma}\label{lem:robust-vision-candidates2}
	For any $g\in P$, (i) $|\hat{Q}(g)|=O(\alpha^{-2})$, and (ii) $\VP^\alpha_{(1,0)}(g)\subseteq\cup_{q\in \hat{Q}(g)}\VP^{\alpha}_{(1/16,0)}(q)$.
\end{lemma}

The proof for \Cref{lem:robust-vision-candidates2} follows the lines of the proof for \Cref{lem:robust-vision-candidates}, with one change - instead of using \Cref{lem:cone_intersecting_disk} we use \Cref{lem:cone_intersecting_disk2}, which we prove below.

\begin{lemma}\label{lem:cone_intersecting_disk2}
    Let $K$ be a cone defined by two rays $\rho_0,\rho_1$ originated at $p$ with small angle $\theta$. If both $\rho_0,\rho_1$ intersect $D_v$, and $p$ sees $K\cap D_v$, then there exists a point in $\hat{Q}_v$ that $(1/16,0,\alpha)$-robustly guards $p$.
\end{lemma}
\begin{proof}
    Assume for simplicity that $\rho_0$ lies on the $x$-axis, and $\rho_1$ lies above it. For $0<\gamma<1$, denote by $\rho_{\gamma}$ the ray from $p$ between $\rho_0$ and $\rho_1$ with angle $\gamma\cdot\theta$ from $\rho_0$.
    First, by \Cref{obs:disk_in_cone} for any point $q$ that lies in $K$ between $\rho_{1/4}$ and $\rho_{3/4}$, the disk $D(q,\frac{\alpha}{4}|p-q|)$ is contained in $K$. In addition, if $v$ lies in $K$, then since both $\rho_0,\rho_1$ intersect $D_v$, we get that $R_v$ is at least the distance between $v$ and one of $\rho_0,\rho_1$. As in the proof of \Cref{obs:disk_in_cone}, we get that $R_v\ge \sin\frac\theta2\|p-v\|>\frac{\sin\theta}{2}\|p-v\|=\frac\alpha2\|p-v\|$. Therefore, if $v$ lies between $\rho_{1/4}$ and $\rho_{3/4}$, then clearly $D(v,\frac{\alpha}{4}\|p-v\|)$ is contained in $K\cap D_v$. 
    We thus assume w.l.o.g. that $v$ lies above $\rho_{3/4}$ (the case when $v$ lies below $\rho_{1/4}$ is symmetric).

    \begin{figure}[h]
    \centering
    \includegraphics[scale=1.2,page=1]{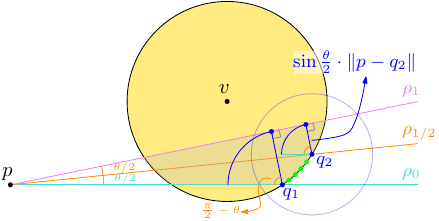}
    %\caption{An illustration for the proof of \Cref{lem:cone_intersecting_disk2}.}    \label{fig:cone_intersecting_disk2}
    \end{figure}
    
    Let $q_1$ (resp. $q_2$) be the point where $\rho_0$ (resp. $\rho_{1/2}$) exits $D_v$.% (see \Cref{fig:cone_intersecting_disk2}). 
    The length of the arc of $D_v$ that subtends an angle $\theta/2$ at the boundary of $D_v$ is $\theta R_v\ge \alpha R_v>\frac{\alpha}{c}2\pi R_v$ for $c\ge2\pi$. Since the angle between $\rho_{1/2}$ and $\rho_0$ is $\theta/2$, the length of the arc that $p$ sees between $q_1$ and $q_2$ is at least $\frac{\alpha}{c}2\pi R_v$, and thus $p$ must see a candidate point $q$ on this arc.

    Consider the segment $\overline{qw}$ from $q$ perpendicular to $\rho_{1}$.
    Since $v$ is above $\rho_{3/4}$, $w$ lies in $D_v$. As $q$ moves from $q_1$ to $q_2$, the length of $\overline{qw}$ becomes smaller. Consider the ``pizza slice'' $Z$ in the disk centered at $q$ with radius $\|q-w\|$, defined by spanning an angle of $\frac\pi2-\theta$ from $w$ in the direction of $p$. Thus $Z$ is above $\rho_0$, and it is contained in $K\cap D_v$, because $w$ is in $D_v$, and $\rho_1$ is tangent to $D(q,\|q-w\|)$. We have $\|q-w\|=\sin\frac\theta2\cdot\|p-q\|$, and thus the area of $Z$ is
    $$\frac{\frac\pi2-\theta}{2\pi} \pi \|q-w\|^2=\frac{\frac\pi2-\theta}{2\pi} \pi (\sin\frac\theta2\cdot\|p-q\|)^2=\pi\|p-q\|^2\cdot (\frac14-\frac{\theta}{2\pi})\sin^2\frac{\theta}{2}.$$

    The area of $D(q,\alpha\|p-q\|)$ is $\pi\alpha^2\|p-q\|^2=\pi\sin^2\theta\cdot\|p-q\|^2$, and thus the fraction of area that $p$ sees in $D(q,\alpha\|p-q\|)$ is\footnote{For the first equation we used the following trigonometric identities: $\sin^2\frac{\theta}{2}=\frac{1-\cos\theta}{2}$, and $\sin^2\theta=1-\cos^2\theta$.}
    $$\frac{\pi\|p-q\|^2\cdot (\frac14-\frac{\theta}{2\pi})\sin^2\frac{\theta}{2}}{\pi\|p-q\|^2\sin^2\theta}=\frac{\pi-2\theta}{4\pi}\cdot\frac12\cdot\frac{1}{1+\cos\theta}=\frac{\pi-2\theta}{8\pi(1+\cos\theta)}.$$
    The Taylor series for $\cos\theta$ is $1-\frac{\theta^2}{2!}+\frac{\theta^4}{4!}-\frac{\theta^6}{6!}+...\le 1-\frac{\theta^2}{2}$, and we get that 
    $$\frac{\pi-2\theta}{8\pi(1+\cos\theta)}\ge \frac{\pi-2\theta}{16\pi-4\pi\theta^2}\ge\frac{1}{16}.$$
\end{proof}

To run \Cref{alg:robust-greedy-approx} in polynomial time, we need an algorithm that computes $\VP^\alpha_{(1/6,0)}$ in polynomial time. We leave this problem as an open question, and obtain the following theorem.
\begin{theorem}\label{thm:refined-robust}
    If $\VP^\alpha_{(1/6,0)}(g)$ can be computed in polynomial time, then a set of $O(\alpha^{-3})|\OPT_\alpha|$ points that $(1/6,0,\alpha/2)$-robustly guard $P$ can be computed in polynomial time.
\end{theorem}

%% file: o-sec_missing_proofs.tex
\section{Missing proofs}\label{sec:proofs}
In this section we provide complete proofs for some of the more technical claims, or slightly different variants of well-known theorems.

\subsection{Proof of the fat-collection lemma}\label{sec:proof-fat-collection}

Let $\F$ be a collection of geometric objects, together with a function $\mathrm{size}:\F\rightarrow \reals_+$. For a parameter $0<\gamma<1$, $\F$ is called a \dfn{$\gamma$-fat collection} if for every object $F\in\F$, there exists a set $H(F)$ of $O(\gamma^{-1})$ points, such that any $F'\in\F$ with $\mathrm{size}(F')\ge\mathrm{size}(F)$ and $F' \cap F \neq \emptyset$ contains a point from $H(F)$. We call $H(F)$ the set of hitting points that correspond to $F$.

It is well-known that if the objects in $\F$ are $\gamma$-fat under an appropriate definition of fatness, where $0 < \gamma < 1$ is some constant and the size of an object is its minimum enclosing ball, then $H(F)$ (as defined above) is of size $O(1)$; see, e.g.~\cite{StappenO94}.
In this paper, we need the following version of a fat-collection theorem, for $\alpha$-fat star-shaped objects in the plane. Recall that we define the \dfn{size} of a star-shaped object $P$ with respect to a given center point $o$ to be the radius of its minimum enclosing ball centered at $o$. 

Recall \Cref{lem:fat-collection}:

\fatCollection*

The proof of the lemma follows from \Cref{clm:hitting-grid} and \Cref{clm:hitting-grid-disk} below.

\begin{claim}\label{clm:hitting-grid}
	Let $G$ be a grid of side length $\frac{\alpha R}{c}$ ($c$ will be determined later). Any $\alpha$-fat star-shaped polygon of size at least $R$ contains a point from $G$.
\end{claim}
\noindent\emph{Proof.}
	Let $S$ be an $\alpha$-fat star-shaped polygon of size at least $R$, and denote by $p$ the center of $S$. Let $q$ be a point in $S$ on the boundary of $D(R,p)$ (such a point exists because $S$ has size at least $R$). Since $S$ is star shaped, $\overline{pq}\subseteq P$. Assume by contradiction that $S$ does not contain any point from $G$.
 
\begin{wrapfigure}{r}{0.4\textwidth}
	\centering
 \vspace{-20pt}
	\includegraphics[scale=1]{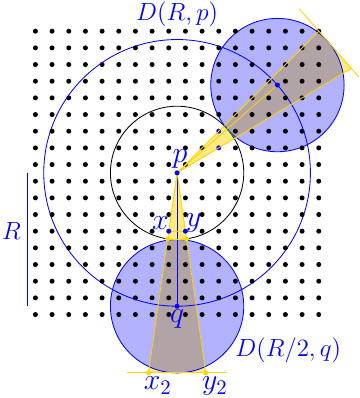}
 \vspace{-20pt}
\end{wrapfigure}
	Let $x,y\in G$ be two adjacent points in $G$ (i.e. on the same grid edge) such that $x,y\in D(R/2,p)$, the segment $\overline{pq}$ cross the edge $xy$, and $x$ is the closest such grid point to $q$.

	The two rays from $p$ to $x$ and $y$ and the tangent to $D(R/2,q)$ in the intersection point with $\overline{pq}$ are forming a triangle $\triangle px_1y_1$ of height $\frac{R}{2}$.
    Since $\|x-y\|=\frac{\alpha R}{c}$, we have by Thales's theorem
	$$\|x_1-y_1\|\le \frac{\frac{R}{2}\cdot \|x-y\|}{\frac{R}{2}-\frac{\alpha R}{c}}\le \frac{\frac{R}{2}\cdot\frac{\alpha R}{c}}{\frac{R}{2}-\frac{\alpha R}{c}}=\frac{\alpha R}{c-2\alpha},$$
	and the area of $\triangle px_1y_1$ is $$R/2\cdot\|x_1-y_1\|=\frac{1}{2c-4\alpha}\alpha R^2.$$
	
	Now consider the larger triangle $\triangle px_2y_2$ defined by the intersection of the same rays from $p$ and the tangent to $D(R/2,q)$ in the second intersection point with the line through $p,q$. The height of this triangle is $\frac{3R}{2}$, and by Thales's theorem we have that 
	$$\|x_2-y_2\|=\frac{\frac32 R\cdot \|x_1-y_1\|}{R/2}\le \frac{3\alpha R}{c-2\alpha}.$$ 
	Thus, the area of $\triangle px_2y_2$ is at most 
	$$\frac32 R\cdot \|x_2-y_2\|=\frac{9}{2c-4\alpha}\alpha R^2.$$
	
	Consider the set $Q$ of points $q'\in D(R/2,q)$ such that $\overline{pq'}$ is between $x$ and $y$, i.e., all points in $D(R/2,q)$ that $p$ sees to the right of $x$ and to the left of $y$.
	Notice that since $S$ does not contain both $x$ and $y$, all points in the connected component of $D(R/2,q)$ and $S$ that contains $q$ are in $Q$. Clearly, $Q$ is contained in $\triangle px_2y_2\setminus \triangle px_1y_1$, and therefore the area of the connected component of $D(R/2,q)$ and $S$ that contains $q$ is 
	$$\frac{9}{2c-4\alpha}-\frac{1}{2c-4\alpha}\cdot\alpha R^2=\frac{8}{2c-4\alpha}\cdot\alpha R^2,$$
	which is smaller than $\alpha\pi(\frac{R}{2})^2$ for $c=6$ and $\alpha\le\frac14$, in contradiction to the fatness of $S$.	
\qed

\begin{claim}\label{clm:hitting-grid-disk}
	Any $\alpha$-fat star-shaped polygon of size at least $R$ that intersects $D(R,o)$, contains a point from $G\cap D(4R,o)$.
\end{claim}
\begin{proof}
	Let $S$ be an $\alpha$-fat star-shaped polygon of size at least $R$, and denote by $p$ the center of $S$. If $S$ intersect $D$, then $p$ must lie in $D(3R,o)$. Moreover, there exists a point $q$ in $S$ on the boundary of $D(R,p)$. We get that $D(R,p)\subseteq D(4R,o)$, and thus by applying the same arguments from \Cref{clm:hitting-grid} we get that $S$ contains a point from $G\cap D(4R,o)$.
\end{proof}

\Cref{lem:fat-collection} now follows, as $G\cap D$ contains $O(\alpha^{-2})$ points (in the bounding box of $D$ there are $(\frac{6}{\alpha})^2$ points).

\subsection{Additional proofs from \Cref{sec:robust-guarding}}\label{sec:proofs-robust-guarding}

Recall \Cref{obs:disk_in_cone}:
\diskInCone*
\begin{proof}
    Assume that the angle between $\overline{pq}$ and the ray $\rho_0$ is $c\cdot\theta$ (and it is smaller or equal to the angle between $\overline{pq}$ and the ray $\rho_1$). Denote by $r$ the distance from $q$ to $\rho_0$ (see \Cref{fig:disk_in_cone}). Then $r=\sin(c\theta)\|p-q\|\ge c\cdot\sin\theta\|p-q\|=c\alpha\|p-q\|$. (The inequality $\sin(cx)>c\cdot\sin(x)$ can be verified by considering the Maclaurin series of $\sin(x)$.) Thus the disk $D(q,c\alpha\|p-q\|)$ is contained in $K$.    
\end{proof}

\subsubsection{Proofs for \Cref{sec:robust-visibility-region} --- The robust visibility region}\label{sec:proofs_robust_vis}
Recall \Cref{clm:union_kites_fat}:
\unionKitesFat*

    \begin{figure}[h]
	\centering
	\includegraphics[scale=0.8]{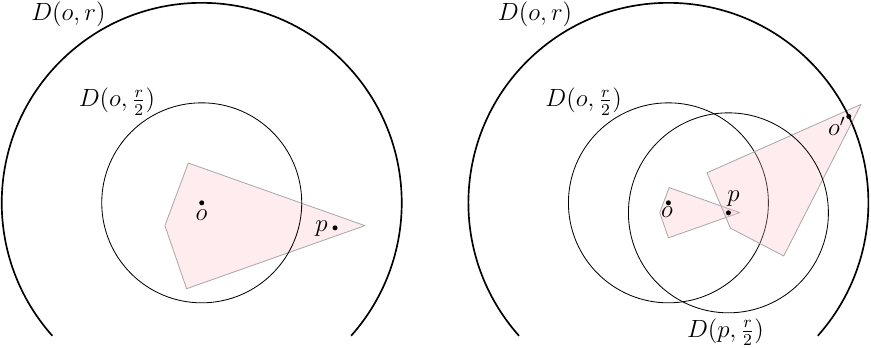}
	\caption{\label{fig:union_kites_fat} The two cases in the proof of \Cref{clm:union_kites_fat}.}
    \end{figure}
    
\begin{proof}
    Let $o$ be a point in $U$, and let $K\in \K$ be a kite containing $o$. Let $D(o,r)$ be a disk that does not fully contain $U$. 

    If $\frac{r}{2}\le \|o-p\|$, then $\partial D(o,\frac{r}{2})$ intersects $K$ (see \Cref{fig:union_kites_fat}, left). Since $K$ is $\gamma$-fat, the connected component of $K\cap D(o,\frac{r}{2})$ that contains $o$ has area at least $\gamma\pi\frac{r^2}{4}$, which is a $\frac{\gamma}{4}$-fraction of the area of $D(o,r)$.
    
    If $\frac{r}{2}>\|o-p\|$, then, since $D(o,r)$ does not fully contain $U$, $\partial D(o,r)$ must intersect $U$; let $o'$ be a point in $\partial D \cap U$, and let $K'\in \K$ be a kite containing $o'$ (see \Cref{fig:union_kites_fat}, right). Consider the disk $D(p,\frac{r}{2})$. Since $\frac{r}{2}>\|o-p\|$, we have $p\in D(o,\frac{r}{2})$, and therefore $D(p,\frac{r}{2})\subseteq D(o,r)$. In addition, $\|p-o'\|>\frac{r}{2}$ because $o'\in\partial D(o,r)$, and thus $D(p,\frac{r}{2})$ intersects $K$'. Since $K'$ is $\gamma$-fat, the connected component of $K'\cap D(p,\frac{r}{2})$ that contains $p$ (which is contained in the connected component of $U\cap D(o,r)$ that contains $o$) has area at least $\gamma\pi\frac{r^2}{4}$, which is a $\frac{\gamma}{4}$-fraction of the area of $D(o,r)$.
\end{proof}

Recall \Cref{clm:computing_vis_1_0}:
\computingVis*

\noindent\emph{Proof.}
    First, notice that by definition $\VP_\alpha(g)\subseteq\VP(g)$, and $\VP(g)$ can be computed in $O(n\log n)$ time. We therefore assume in this proof that $P=\VP(g)$.
  
    Let $v\in P$ be a reflex vertex, and assume w.l.o.g. that both $g,v$ are on the $x$-axis, and $g$ is to the left of $v$ (see the figure to the right, top). By \Cref{obs:p_is_close}, any point in $\VP_\alpha(g)$ is within distance $\frac{\|g-v\|}{\alpha}$ of $g$, and therefore $\VP_\alpha(g)\subseteq D(g,\frac{\|g-v\|}{\alpha})$.

    \begin{figure}[h]
        \centering
        \includegraphics[scale=0.5]{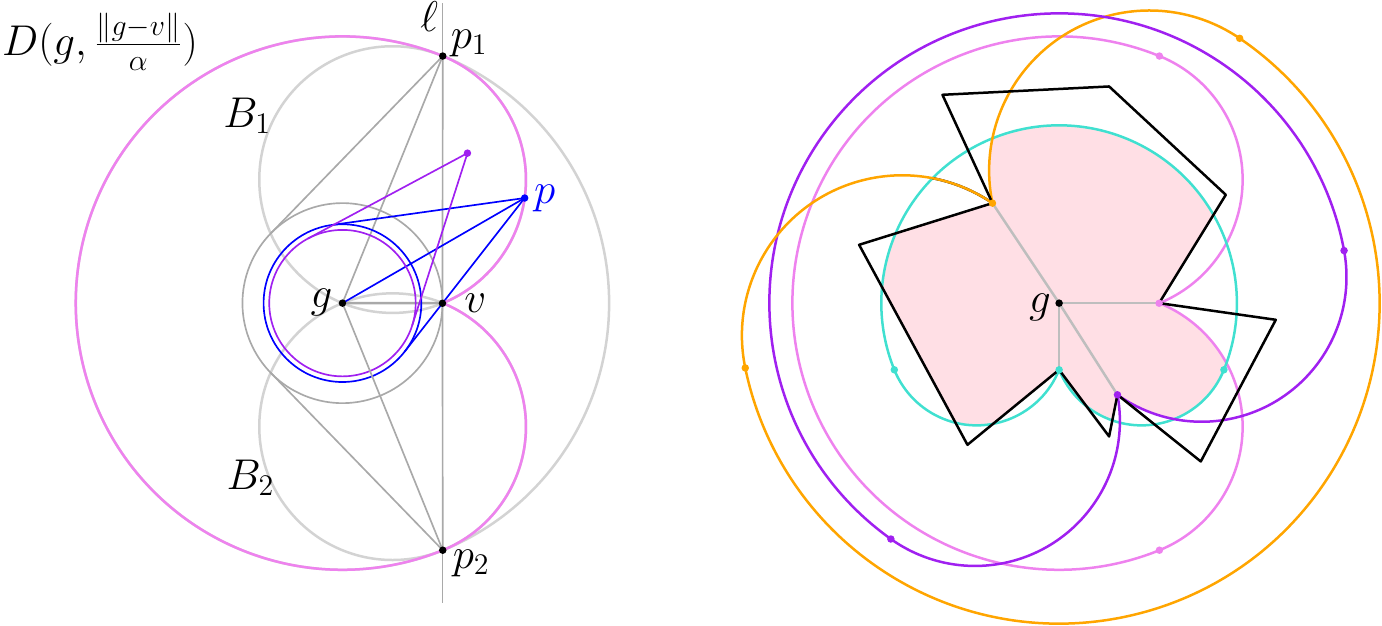}
    \caption{Computing $\VP_\alpha(g)$.}
    \label{fig:compute-vis}
    \end{figure}
    
    Let $\ell$ be the vertical line through $v$ (see \Cref{fig:compute-vis}, left), and denote by $p_1,p_2$ the two intersection points of $\ell$ and $D(g,\frac{\|g-v\|}{\alpha})$ (so $\|p_1-g\|=\|p_2-g\|=\frac{\|g-v\|}{\alpha}$).
    The line $\ell$ is tangent to the disk $D(g,\|g-v\|)$ 
    at the point $v$, and we have $\|g-v\|=\alpha\|p_1-g\|$. Since $\angle gvp_1=\frac{\pi}{2}$, we have $\angle gp_1v=\theta$. Symmetrically, $\angle gp_2v=\theta$.
    Denote the part of $D(g,\frac{\|g-v\|}{\alpha})$ to the left of $\ell$ by $C_0$.
    Notice that for any point $p\in C_0$, the ice cream cone from $p$ to $g$ is contained in $C_0$.
    
    Let $B_1$ (resp. $B_2$) be the disk through $g,v,p_1$ (resp. $g,v,p_2$). Let $p$ be a point on the arc of $B_1$ between $p_1$ and $v$ (that does not contain $g$). We have $\angle gpv=\theta$, because $\angle gp_1v=\theta$ and both angles are subtended by the same arc of $B_1$. Since $\angle gpv=\theta$, the ray from $p$ in the direction of $v$ is tangent to $D(g,\alpha\|p-g\|)$. This means that for any such point $p$, the reflex vertex $v$ is almost ``blocking'' the view for $p$, so that any point to the right of this arc is not $\alpha$-robustly visible to $g$. More precisely, for any point $p$ on one of the arcs $\overset{\frown}{p_1v},\overset{\frown}{p_2v}$, the ice cream cone from $p$ to $g$ exactly touches $v$. 
    Let $C_1$ (resp. $C_2$) be the part of $B_1$ (resp. $B_2$) to the right of $\ell$. Then, for any point $p\in C_0\cup C_1\cup C_2$ the respective ice cream cone is contained in $C_0\cup C_1\cup C_2$, and for any point $p\notin C_0\cup C_1\cup C_2$ the respective ice cream cone contains $v$. We conclude that $\VP_\alpha(g)$ is contained in $C_0\cup C_1\cup C_2$, and that if not other features of $P$ are contained in $C_0\cup C_1\cup C_2$, then $\VP_\alpha(g)=C_0\cup C_1\cup C_2$.

    Therefore, we can compute $\VP_\alpha(g)$ in $O(n^2)$ time as follows. For each reflex vertex, compute the region bounded by the three arcs defining $C_0\cup C_1\cup C_2$. In addition, compute $R_g$ (note that it is possible that $R_g$ is the distance from $g$ to an edge of $P$). Compute the intersection between all those regions, the disk $D(g,\frac{R_g}{\alpha})$, and the polygon $P$. The total number of intersection points is $O(n^2)$, and thus we can compute the intersection of all these regions in $O(n^2)$ time; see \Cref{fig:compute-vis}, right.
\qed

\subsubsection{Proofs for \Cref{sec:inv_vis} --- The robust inverse visibility region}\label{sec:proofs_inv_vis}
Recall \Cref{clm:VP_0_1_fat_kite}:
\VPFatKite*

\begin{proof}
Consider the ice cream cone from $p$ to $g$, and assume w.l.o.g. that $p,g$ are on the $x$-axis. Let $\rho^a,\rho^b$ be the upper and lower tangents from $p$ to $D(g,\alpha\|g-p\|)$, and let $a,b$ be the tangency points (see \Cref{fig:line_visibility}, right). Let $\rho^a_{1/2}$ (resp. $\rho^b_{1/2}$) be the bisector of $\angle apg$ (resp. $\angle bpg$). By \Cref{obs:disk_in_cone}, for any point $q$ in the cone of $\rho^a_{1/2},\rho^b_{1/2}$, the disk $D(q,\frac\alpha2\|p-q\|)$ is contained in the cone of $\rho^a,\rho^b$.

Denote by $a'$ (resp. $b'$) the intersection point between $\overline{ga}$ (resp. $\overline{gb}$) and $\rho^a_{1/2}$ (resp. $\rho^b_{1/2}$). Since $\angle pag=\frac{\pi}{2}$, for any point $q\in\overline{ga'}$, the disk $D(q,\frac\alpha2\|p-q\|)$ is contained in $D(g,\alpha\|p-g\|)$ (because by \Cref{obs:disk_in_cone}, the upper tangent to $D(q,\frac\alpha2\|p-q\|)$ must be below $\rho^a$).
Consider the kite $K=ga'pb'$; then the angle at $p$ is $\angle a'pg+\angle b'pg=\frac{\theta}{2}+\frac{\theta}{2}=\theta$, and the angle at $g$ is $\pi-2\theta$. Since for any point $q\in K$, the disk $D(q,\frac\alpha2\|p-q\|)$ is contained in the ice cream cone from $p$ to $g$, we get that every $q\in K$ $\alpha/2$-robustly guards $p$.
Moreover, since the 4 angles of the kite are at least $\theta=\arcsin\alpha$, the kite is $O(\alpha)$-fat.
\end{proof}

Recall \Cref{lem:approx-alpha}:
\approxAlpha*

\begin{proof}
    For every point $g\in \VP^{\text{inv}}_\alpha(p)$, let $K(g)$ be the $O(\alpha)$-fat kite from \Cref{clm:VP_0_1_fat_kite}. Then $K(g)$ has $p$ and $g$ as its vertices, and any point $q\in K(g)$ $\alpha/2$-robustly guards $p$. Let $F_p$ be the union of all these kites, i.e., $\bigcup_{g\in  \VP^{\text{inv}}_\alpha(p)}K(g)$. Clearly, $F_p$ is star-shaped, it contains $\VP^{\text{inv}}_\alpha(p)$, and by \Cref{clm:union_kites_fat}, it is $O(\alpha)$-fat. The size (radius of the smallest enclosing disk centered at $p$) of $F_p$ is equal to the size of $\VP^{\text{inv}}_\alpha(p)$, because $F_p$ contains $\VP^{\text{inv}}_\alpha(p)$, and for any $g\in \VP^{\text{inv}}_\alpha(p)$, the furthest point from $p$ in a kite $K(g)$ is $g$.
\end{proof}

Recall \Cref{clm:computing_vis_0_1}:
\computingVisInv*

\begin{proof}
	By definition we have $\VP^{\text{inv}}_\alpha(p)\subseteq\VP(g)$, so we assume in this proof that $P=\VP(g)$.
	
	Given a point $g\in P$, recall that $R_g$ is the largest radius of a disk centered at $g$ and contained in $P$. Since $p$ can see any point in $P$, $p$ is $\alpha$-robustly guarded by $g$ if and only if $R_g\ge\alpha\|p-g\|$. Notice that $R_g$ is determined by the distance from $g$ to the closest edge or vertex of $P$.
	
	For any edge $\{u,v\}$ of $P$, if the distance from $g$ to $\overline{uv}$ is smaller than $\alpha\|p-g\|$, then $R_g<\alpha\|p-g\|$, and $p$ is not $\alpha)$-robustly guarded by $g$.	Therefore, for each edge of $P$ we compute the locus of all points $g\in P$ such that $\alpha\|p-g\|$ is equal to the distance between $g$ and $\overline{uv}$, as follows.
	
	To make the computations easier, assume w.l.o.g. (by scaling and rotation) that $u,v$ are on the $x$-axis, $u$ is to the right of $v$, and that $p=(0,1)$ (as illustrated in the figure below). 
 \begin{figure}[htb]
	\centering
	\includegraphics[scale=0.9]{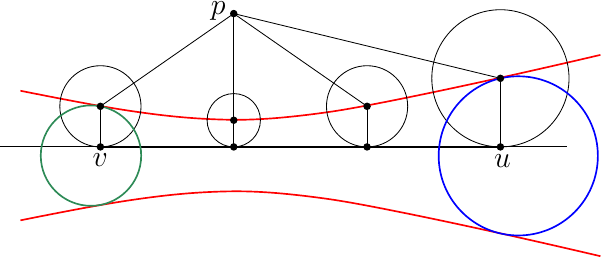}
\end{figure}

    The distance between a point $g=(x_g,y_g)$ and the $x$-axis is simply $y_g$. We have $\|p-g\|=\sqrt{x_g^2+(1-y_g)^2}$, and thus the distance between $g$ and the $x$-axis is exactly $\alpha\|p-g\|$ for any point $g$ on the following hyperbola $H_e$ (colored red in the figure above):
	\begin{align*}
	&y_g=\alpha\sqrt{x_g^2+(1-y_g)^2}\\
	&\frac{y_g^2}{\alpha^2}=x_g^2+(1-y_g)^2=x_g^2+1-2y_g+y_g^2\\
	&x_g^2+(1-\alpha^{-2})y_g^2-2y_g+1=0
	\end{align*}
	
	The distance between $g$ and $v=(x_v,0)$ is $\sqrt{(x_v-x_g)^2+y_g^2}$, and it is equal to $\alpha\|p-g\|$ for any $g$ on the following disk (colored green in the figure above):
	\begin{align*}
	&\sqrt{(x_v-x_g)^2+y_g^2}=\alpha\sqrt{x_g^2+(1-y_g)^2}\\
	&(x_v-x_g)^2+y_g^2=\alpha^2 (x_g^2+(1-y_g)^2)\\
	&x_g^2+x_v^2-2x_vx_g+y_p^2=\alpha^2(x_g^2+1-2y_g+y_g^2)\\
	&(1-\alpha^2)x_g^2-2x_vx_g+(1-\alpha^2)y_g^2+2\alpha^2y_g-\alpha^2+x_v^2=0
	\end{align*}
	 We compute the disk corresponding to $u$ similarly (colored green in the figure above). 
    Denote by $C_v$ the disk corresponding to $v$, and by $C_u$ the disk corresponding to $u$. 
    
    Notice that $H_e,C_v$ intersect at a point $q_v$ directly above $v$, and $H_e,C_u$ intersect at a point $q_u$ directly above $u$. For any point $g\in P$, $D(g,\alpha\|p-g\|)$ intersect the edge $\{u,v\}$ if and only if $g$ is between the branches of $H_e$ and between $\overline{uq_u}$ and $\overline{vq_v}$, or in one of the disks $C_u,C_v$. Thus, for each edge of the polygon we have a constant number of constraints, and computing their intersection can be easily done in $O(n^2)$ time.
\end{proof}

\subsubsection{Hardness proof}\label{sec:proof_hardness}

Recall \Cref{thm:hardness}:
\hardness*

\begin{proof}
We do an approximation preserving reduction from the {\sc Hitting Lines} problem~\cite{broden2001guarding}:
Given a finite set ${\cal L}$ of lines in the plane, no two of which are parallel, find a minimum-cardinality set of points $H$ that hit all lines of ${\cal L}$ (i.e., there is at least one point of $H$ on each line of ${\cal L}$). 

\begin{figure}[h]
	\centering
	\includegraphics[scale=0.7]{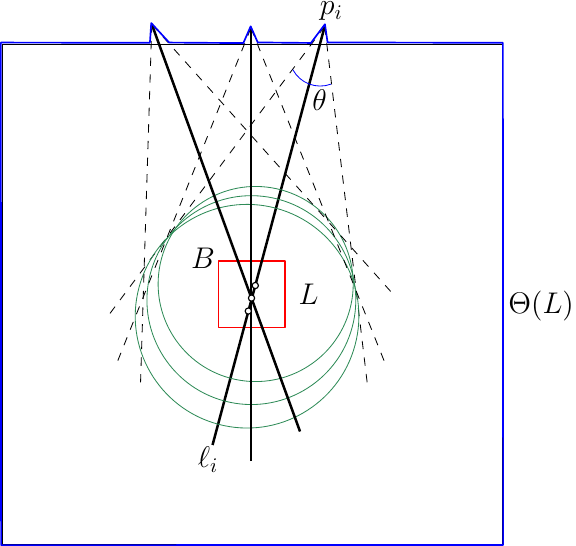}
	\caption{\label{fig:hitting-lines} Reducing an instance of Hitting Lines to an instance of robust guarding in the blue simple polygon (a ``spike box'').}
\end{figure}

Given the set ${\cal L}$ of lines, we construct a simple polygon $P$ (a ``spike box'') that consists of a box bounding all vertices in the arrangement of the lines ${\cal L}$, with triangular ``spikes'' extending out from it. Refer to Figure~\ref{fig:hitting-lines}. The interior angle, $2\theta$, at the tip (convex) vertex of each spike is chosen in accordance with the choice of robustness parameter $\alpha$, so that the only points within $P$ that can $\alpha$-robustly guard each tip $p_i$ are those points that lie on (hit) the corresponding line, $\ell_i$, that is the bisector of the cone defined by extending the edges incident on $p_i$. (In the figure, $\alpha$ is approximately 0.4.) Thus, there is a solution of size $k$ to the {\sc Hitting Set} instance ${\cal L}$ if and only if the polygon $P$ can be $\alpha$-robustly guarded by $k$ guard points.

One can check that the instance can be constructed with the necessary precision using a polynomial number of bits. In particular, each line of ${\cal L}$ can be assumed to be given as a pair of points from an integer grid, $[0,N]\times[0,N]$. Then, the smallest angle between any two lines is at least $\Theta(1/N^2)$, implying that all vertices of the arrangement lie within a box (shown in red in the figure) of side length $L=\Theta(N^3)$. The larger enclosing box that forms the basis of the spike box can then be taken of size $\Theta(L)$, and all vertices of $P$ can be chosen to be rational from a comparably sized grid.
\end{proof}

\subsection{Additional proofs from \Cref{sec:robust-guarding-algorithm}}\label{sec:proofs-const-approx}

Recall \Cref{obs:purple_regions}:

\purpleRegions*

\begin{proof}
If $g\in \Pi$,
we show that $\VP_\alpha(g)$ intersects at most four of the disks $D_{p_i}$ that were added to $\Pi$. 
Draw a perpendicular from $g$ onto the medial axis; let $h$ be the intersection point on the medial axis. Suppose $h$ is in between the two consecutive disk-centers $p_i$ and $p_{i+1}$. 
If $g$ sees a point $\alpha$-robustly, then the point is also seen $\alpha$-robustly by either $p_i$ or $p_{i+1}$. Hence, $\VP_\alpha(g)\subseteq \VP_\alpha(p_i)\cup \VP_\alpha(p_{i+1})$. Now $\VP_\alpha(p_i)$ can intersect three disks centered at $p_i$ and $p_{i-1}$ and $p_{i+1}$. Similarly, $\VP_\alpha(p_{i+1})$ can intersect three disks centered at $p_{i+1}$ and $p_{i}$ and $p_{i+2}$. Hence, $\VP_\alpha(g)$ can intersect at most two disks to the left of $h$ (disks centered at $p_i$ and $p_{i-1}$) and at most two disks to the right of $h$ (disks centered at $p_{i+1}$ and $p_{i+2}$).

If $g\notin \Pi$, then we will prove this observation by contradiction. Suppose that $\VP_\alpha(g)$ intersects more than 4 disks in $\Pi$. In that case there must be a point $p$ in $\Pi$ such that $p\in\VP_\alpha(g)$ and $\overline{gp}$ intersects more than 4 disks in $\Pi$. Then there must be another point $g'\in \Pi$ such that $g'$ $\alpha$-robustly sees $p$ and thus $\VP_\alpha(g')$ intersects more than 4 disks in $\Pi$, which is a contradiction. The observation is thus proved. 
\end{proof}

Recall \Cref{obs:boundary-cells}:

\boundaryCells*

\begin{proof}
    To prove the observation, we show that any blue region is bounded by two non-disjoint medial disks and one edge of $P$. 

    Consider a blue region $R$ (see \Cref{fig:medial_axis_cells}). Since $R\subseteq P\setminus \D$, it must be bounded by an edge $e\in P$ and a disk $D_v$ that touches $e$. Let $\{v,u\}$ be the medial axis edge such that $R$ is between $D_u$ and $D_v$. If $D_u\cap D_v=\emptyset$, then $\{v,u\}$ must be a straight-line edge (otherwise, if it were a subcurve of a parabola, both $D_v$ and $D_u$ would touch the reflex vertex defining this parabola). In this case, the region of the polygon between $D_u$ and $D_v$ is purple. We get that $D_u$ and $D_v$ are non-disjoint.
\end{proof}

\subsubsection{Additional proofs for \Cref{sec:robust-candidate-guards}}\label{sec:proofs_candidates}

The following claim completes the proof of \Cref{lem:cone_intersecting_disk}.

\begin{claim} \label{claim:quad}
Quadrilateral $z_1y_1y_2 z_2$ must contain a grid point from $Q_v$.
\end{claim}

\begin{proof}
We show that quadrilateral $z_1y_1y_2 z_2$ contains a disk of diameter $\Theta(\alpha^2)\cdot R_v$, and thus must contain a grid point from $Q_v$.

The $\triangle px_2x_3$ is an isosceles triangle. Since $\overline{p w_1}$ is the angle bisector of $\angle x_2 p x_3$ and also a perpendicular to base $\overline{x_2 x_3}$, we get $|\overline{x_2 w_1}| = |\overline{w_1 x_3}|$. This gives $|\overline{x_2 x_3}| = 2\cdot |\overline{x_2w_1}| = \Theta(\theta^2)\cdot R_v$.
From the construction of $\rho_{1/4}, \rho_{1/2}, \rho_{3/4}$ and $\rho_1$, we get $|\overline{w_2 w_3}| = 1/3 \cdot |\overline{x_2 x_3}| = \Theta(\theta^2)\cdot R_v$.

Consider $\triangle p z_1 z_2$ and $\triangle p w_2 w_3$. We get the following ratio between the edges of both triangles.
$$\frac{|\overline{z_1 z_2}|}{|\overline{w_2 w_3}|} = \frac{|\overline{p w_4}|}{|\overline{p w_1}|} = \frac{|\overline{p o}| - |\overline{o w_4}|}{|\overline{p o}| + |\overline{o w_4}|} = \frac{1 - \alpha/4}{1+ \alpha/4}>\frac12.$$    

Thus, $|\overline{z_1 z_2}| >\frac12\cdot |\overline{w_2 w_3}| = \Theta(\theta^2)\cdot R_v$. 

From the construction of quadrilateral $z_1y_1y_2 z_2$, it is a trapezoid (both $\overline{y_1 y_2}$ and $\overline{z_1 z_2}$ are perpendicular to $\overline{po}$, and thus $\overline{y_1 y_2}$ and $\overline{z_1 z_2}$ are parallel). The height of trapezoid $z_1 y_1 y_2 z_2$ is the radius of disk $D(o,\alpha/4\|o-p\|)$, which is $\alpha/4 \cdot |\overline{po}| = \Theta(\alpha) \cdot |\overline{pw_1}| = \Theta(\alpha \cdot \theta)\cdot R_v = \Theta(\alpha^2)\cdot R_v$.

Since the length of the smallest side and the height of trapezoid $z_1 y_1 y_2 z_2$ are $\Theta(\alpha^2)\cdot R_v$, the trapezoid can contain a disk of diameter $\Theta(\alpha^2)\cdot R_v$, and thus must contain a grid point from $Q_v$. The claim is thus proved.
\end{proof}

The following claim completes the proof of \Cref{lem:robust-vision-candidates}.

\begin{claim}\label{clm:angle_is_large}
	$\angle w_1pw_2\ge\theta$.
\end{claim}

\begin{figure}[h]
	\centering
	\includegraphics[scale=1]{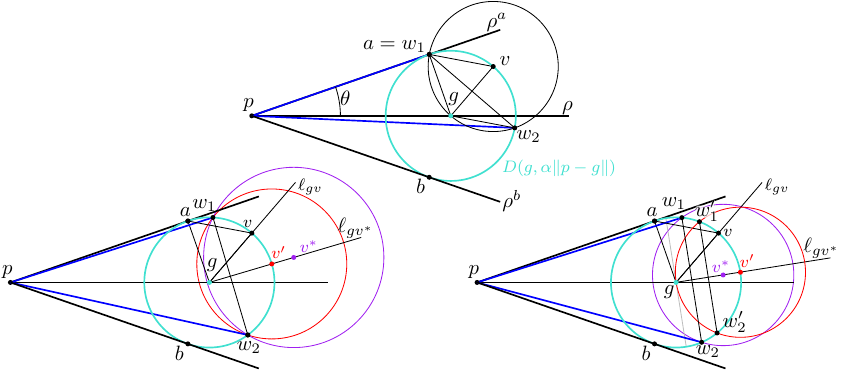}
	\caption{An illustration of the proof for \Cref{clm:angle_is_large}. Top: $v$ is on $\partial D(g,\alpha\|p-g\|)$, $R_v=\alpha\|p-g\|$, and $w_1=a$.
 Bottom left: $v^*\notin D(g,\alpha\|p-g\|)$. Bottom right: $v^*\in D(g,\alpha\|p-g\|)$.}
	\label{fig:min_angle}
\end{figure}

\begin{proof}
     Consider the triangle $\triangle gvw_1$ (see \Cref{fig:min_angle}). We have $\|g-w_1\|=\alpha\|p-g\|$, $\|v-w_1\|=R_v$, and since $g\in D_v$, $\|g-v|\le R_v$. Thus, $\angle w_1 gv$ is opposite the largest side of $\triangle gvw_1$, and it cannot be the smallest of the 3 angles of the triangle, so we know that $\angle w_1 gv\ge\frac\pi3$. Similarly, $\angle w_2 gv\ge\frac\pi3$, and thus $\angle w_1 gw_2\ge\frac23\pi$. We have $\theta\le\frac\pi6$ (for $\alpha\le\frac12$), and since $w_1$ is below $\rho^a$, we have $\angle pw_1g\le\angle pag=\frac\pi2$. We get that $\angle pgw_1\ge \frac\pi3$, and thus $w_2$ must be on or below $\rho$.
    
    Now assume that $v$ is on $\partial D(g,\alpha\|p-g\|)$, $R_v=\alpha\|p-g\|$, and $w_1=a$. We have $\angle w_1pg=\theta$, and $\angle gw_1p=\frac\pi2$ (as $w_1=a$, the tangent point of $\rho^a$ and $D(g,\alpha\|p-g\|)$). Since $w_2$ is on or below $\rho$, $w_1=a$, and $\angle w_1pg=\theta$, we have $\angle w_1pw_2\ge\theta$.

    We now claim that the above case gives the lower bound. In other words, the angle $\angle w_1pw_2$ is minimized when $v$ is on $\partial D(g,\alpha\|p-g\|)$, $R_v=\alpha\|p-g\|$, and $w_1$ is (almost) on $a$. Consider a disk that minimizes the angle $\angle w_1pw_2$, and let $D(v^*,r^*)$ be such disk with \textbf{minimum} radius $r$. Recall that $D(v^*,r^*)$ contains $g$ but not $a$, $r^*\ge\alpha\|p-g\|$, and $v^*$ is above $\rho$.

    Let $v$ be the point on $\partial D(g,\alpha\|p-g\|)$ at distance $\alpha\|p-g\|$ from both $g$ and $a$, and denote by $\ell_{gv}$ the line through $g$ and $v$. Notice that if $v^*$ is above $\ell_{gv}$, then since $D(v^*,r^*)$ contains $g$ it must also contain $a$. Hence, $v^*$ is below $\ell_{gv}$. The line $\ell_{gv^*}$ through $g$ and $v^*$ crosses $\partial D(g,\alpha\|p-g\|)$ at a point $v'$ below $\ell_{gv}$ and above $\rho$.

    If $v^*\notin D(g,\alpha\|p-g\|)$, then $r^*=\|v^*-w_1\|>\|v'-w_1\|$ (because $\overline{gv^*}$ is bisecting $\overline{w_1w_2}$).
    Construct a disk $D(v',r')$ such that $r'=\|v'-w_1\|$. Thus the intersection points between $\partial D(v',r')$ and $\partial D(g,\alpha\|p-g\|)$ remain $w_1,w_2$. Now, as $\angle w_1 gv'=\angle w_1 gv^*\ge\frac\pi3$, and $\|g-w_1\|=\|g-v'\|$, we get that $r'=\|v'-w_1\|\ge \|g-v'\|=\alpha\|p-g\|$ and thus $D(v',r')$ contains $g$. We get that $D(v',r')$ contains $g$ but not $a$, and $r'\ge\alpha\|p-g\|$. The angle $\angle w_1pw_2$ remains the same, but $r^*>\|v'-w_1\|=r'$, in contradiction to $r^*$ being minimal.
    
    If $v^*\in D(g,\alpha\|p-g\|)$, consider the disk $D(v',\alpha\|p-g\|)$, and let $w'_1,w'_2$ be the two intersection points between $D(v',\alpha\|p-g\|)$ and $D(g,\alpha\|p-g\|)$. Consider the diameter of $D(g,\alpha\|p-g\|)$ perpendicular to $\overline{gv'}$, then $w_1,w_2$ are on the side as $v'$, and moreover, as we slide $v^*$ towards $v'$ and decrease $r^*$, we get $\|w'_1-w'_2\|<\|w_1-w_2\|$ and $\angle w'_1pw'_2<\angle w_1pw_2$. So we can assume that $r^*=\alpha\|p-g\|$ and $v^*=v'$.
    Now, slide $v^*$ along $\partial D(g,\|p-g\|)$ until $w_1=a$, and let $w$ be the new location of $w_2$. 
    As we slide $v^*$ along $\partial D(g,\|p-g\|)$ until $w_1=a$, the length $\|w_1-w_2\|$ remains the same and both endpoints of $\overline{w_1w_2}$ stay on $\partial D(g,\|p-g\|)$. Since the position of $p$ is fixed, $\angle w_1 p w_2$ is minimized when $w_1 = a$. 
\end{proof}